\documentclass[
  aps,
  prd,
  onecolumn,
  superscriptaddress,
  nofootinbib,
  longbibliography,
  floatfix,
  amsmath,
  amssymb,
  10pt
]{revtex4-2}

\usepackage[T1]{fontenc}
\usepackage[utf8]{inputenc}
\usepackage{newtxtext,newtxmath}
\usepackage{microtype}
\usepackage[dvipsnames]{xcolor}
\usepackage{amsmath,amssymb,amsfonts}
\usepackage{mathtools}
\usepackage{bm}
\usepackage{graphicx}
\usepackage{booktabs}
\usepackage{orcidlink}
\usepackage{hyperref}
\usepackage{bookmark}
\usepackage{flafter}
\usepackage{placeins}
\usepackage{capt-of}
\allowdisplaybreaks
\definecolor{darkblue}{rgb}{0,0,0.5}
\hypersetup{
  pdfstartview={FitH},
  pdftitle={A causal magnetic black hole with finite self-energy},
  pdfauthor={Mohsen Fathi, Ariel Guzman, J. R. Villanueva},
  pdfsubject={Nonlinear electrodynamics, magnetic black holes, causal structure, thermodynamics, and optical appearance},
  pdfkeywords={nonlinear electrodynamics, magnetic black holes, finite self-energy, causal structure, photon propagation, black hole imaging},
  colorlinks=true,
  breaklinks=true,
  linkcolor=darkblue,
  citecolor=darkblue,
  urlcolor=darkblue
}

\graphicspath{{./}{figures/}}

\newtheorem{lemma}{Lemma}
\newcommand{\dd}{\mathrm{d}}
\newcommand{\cF}{\mathcal{F}}
\newcommand{\cL}{\mathcal{L}}
\newcommand{\LF}{\mathcal{L}_{\mathcal F}}
\newcommand{\LFF}{\mathcal{L}_{\mathcal F\mathcal F}}
\newcommand{\kappaem}{\kappa_{\rm em}}

\begin{document}

\title{A causal magnetic black hole with finite self-energy}

\author{Mohsen Fathi\orcidlink{0000-0002-1602-0722}}
\email{mohsen.fathi@ucentral.cl}
\affiliation{Centro de Investigaci\'{o}n en Ciencias del Espacio y F\'{i}sica Te\'{o}rica (CICEF), Universidad Central de Chile, La Serena 1710164, Chile}

\author{Ariel Guzm\'an\orcidlink{0009-0008-3844-1203}}
\email{ariel.guzman@estudiantes.uv.cl}
\affiliation{Instituto de F\'{i}sica y Astronom\'{i}a, Facultad de Ciencias, Universidad de Valpara\'{i}so, Avenida Gran Breta\~na 1111, Valpara\'{i}so, Chile}

\author{J. R. Villanueva\orcidlink{0000-0002-6726-492X}}
\email{jose.villanueva@uv.cl}
\affiliation{Instituto de F\'{i}sica y Astronom\'{i}a, Facultad de Ciencias, Universidad de Valpara\'{i}so, Avenida Gran Breta\~na 1111, Valpara\'{i}so, Chile}

\date{\today}

\begin{abstract}
We construct a nonlinear electrodynamics (NLED) model for a static magnetic black hole. Rather than imposing a regular center, we require a Maxwell weak-field limit, finite magnetic self-energy, a positive and subluminal electromagnetic cone, the standard energy conditions, and a static exterior that satisfies known sufficient stability criteria. A positive mixture of power kernels restricts the exponent to $1/4<\gamma\leq1/2$. For the minimal two-kernel model, with $\gamma_1=1/3$ and $\gamma_2=1/2$, the field equations admit an exact solution in terms of incomplete beta functions. For the branches with nonnegative Schwarzschild mass parameter $M_0$, the metric function is strictly increasing, so there is at most one positive-radius horizon and no inner Cauchy horizon. We construct the maximal extension and show that the black-hole, horizonless, and critical branches have spacelike, timelike, and null singularities, respectively. The representative black-hole families have positive temperature and negative fixed-charge heat capacity. Electromagnetic waves split into ordinary and extraordinary optical branches. We calculate their photon spheres, shadow radii, instability rates, and illustrative Event Horizon Telescope size bands. We also construct ISCO-truncated thin-disk images. The full images remain close, but jointly normalized residual maps and radial profiles reveal a coherent branch-dependent shift near the lensed inner edge and the critical region. The extraordinary branch moves the critical curve outward and partly compensates the shadow reduction caused by magnetic charge.

\bigskip

{\noindent{\textit{keywords}}: Nonlinear electrodynamics; Magnetic black holes; Causal structure; Black-hole thermodynamics; Photon propagation; Black-hole imaging}

\end{abstract}

\maketitle

\section{Introduction}\label{sec:intro}

Nonlinear electrodynamics (NLED) replaces the Maxwell Lagrangian by a nonlinear function of the electromagnetic invariants. The subject began with the theories of Born and Born--Infeld and with the quantum-electrodynamic correction of Heisenberg and Euler \cite{Born1933,BornInfeld1934,HeisenbergEuler1936,Plebanski1970,BialynickiBirula1984}. It later developed into a broad framework that includes duality-invariant models, effective field theories, and causal extensions of Maxwell electrodynamics \cite{GibbonsRasheed1995,GibbonsHerdeiro2001,BandosEtAl2020,Sorokin2022,RussoTownsend2024}. In curved spacetime, NLED affects both the geometry, through its stress tensor, and the propagation of electromagnetic disturbances \cite{Boillat1970,NovelloEtAl2000,ObukhovRubilar2002,ShabadUsov2011,Schellstede2016}.

Many black hole solutions have been found in Einstein gravity coupled to NLED. The Born--Infeld family provided early examples with finite electric fields, nontrivial horizon structure, and interesting thermodynamics \cite{Breton2003,FernandoKrug2003,CaiPangWang2004,Dey2004,Fernando2006}. Regular black holes later became an important direction. The solutions of Ay\'on-Beato and Garc\'ia, together with the magnetic interpretation of the Bardeen geometry, showed how NLED can support nonsingular metrics \cite{AyonBeatoGarcia1998,AyonBeatoGarcia1999,AyonBeatoGarcia2000,AyonBeatoGarcia2005,CataldoGarcia2000}. Related constructions include electrically and magnetically charged families, de Sitter cores, reverse-engineered models, and new power-law or Born--Infeld-type theories \cite{Dymnikova2004,Bronnikov2001,Ma2015,FanWang2016,Kruglov2017,BronnikovComment2017,RodriguesSilvaSiqueira2020,Bronnikov2024,AlencarEtAl2024,LiLu2024,ContrerasEtAl2025,BokulicEtAl2024,BokulicEtAl2026,TsudaEtAl2026}. Recent work has also clarified exact rotating NLED sources and alternative first-order formulations \cite{GarciaDiaz2022,AyonBeato2024,AyonBeatoFloresHassaine2024,VerbinEtAl2025,BokulicHerdeiro2025,BarrientosEtAl2025}.

Regularity of the metric, however, is not enough by itself. A one-invariant theory with the Maxwell weak-field limit faces a strong restriction in the electric sector, while regular magnetic centers can develop problems in their characteristic or perturbation sectors \cite{Bronnikov2001,Bronnikov2022,BokulicEtAl2024,BokulicEtAl2026,DeFeliceTsujikawaScalar2025}. A direct perturbative analysis found a Laplacian instability near the center of nonsingular Einstein--NLED black holes \cite{DeFeliceTsujikawa2025}. More general studies have shown that causality imposes strong energy-condition constraints and can exclude regular black holes in important NLED classes \cite{RussoTownsend2024,RussoTownsend2026,AbeEtAl2026}. Stable singular solutions may still exist, which makes it important to distinguish a regular center from a physically controlled exterior \cite{ChenEtAl2026}.

We therefore take a different route. We do not force the center to be regular. Instead, we require a Maxwell limit, finite magnetic self-energy, a positive and subluminal magnetic characteristic, the usual energy conditions, and a static exterior that satisfies known sufficient stability conditions. Finite electromagnetic self-energy has recently been connected with the removal of inner Cauchy horizons \cite{HaleEtAl2026}. Other no-Cauchy-horizon results and exact NLED constructions point in the same direction \cite{AnLiYang2021,LeeMyung2026,PinedoFrolov2026}. Our aim is to obtain a simple theory with a clear exterior sector and a global structure that can be studied analytically.

The thermodynamics and perturbations of NLED black holes have been studied in many settings. The first law, generalized Smarr relations, phase transitions, stability conditions, quasinormal modes, and greybody factors can all depend on the nonlinear couplings \cite{Smarr1973,MorenoSarbach2003,ZhangGao2018,BokulicThermo2021,WangWuYang2019,NomuraYoshidaSoda2020,NomuraYoshida2022,DaghighEtAl2022,BalartFernando2021,CaiMiao2021,HegdeEtAl2025,CroneyEtAl2025,LiangEtAl2026}. In earlier work, we studied quasinormal modes of a different static black hole sourced by a Pleba\'nski-type NLED model in an asymptotically anti-de Sitter spacetime \cite{FathiGuzmanVillanueva2026}. Here we instead introduce an asymptotically flat power-mixture theory and derive its exact magnetic solution, causal structure, thermodynamics, and optical sector.

Light propagation provides another strong test of NLED. Effective photon metrics, vacuum birefringence, light rings, lensing, and polarization-dependent shadow edges have been widely studied \cite{NovelloEtAl2000,PerlickEtAl2015,DePaulaEtAl2023,DePaulaEtAl2026,Escobar2026}. Recent works have examined NLED shadows, thin disks, deflection angles, quasinormal modes, and Event Horizon Telescope (EHT) constraints \cite{OkyayOvgun2022,UniyalPantigOvgun2023,UniyalEtAl2023,UniyalEtAl2024,LambiaseEtAl2025,CimdikerEtAl2026,UktamovEtAl2026}. Related studies have considered rotating shadows, effective-metric images, accretion observables, and strong-field constraints \cite{RazaEtAl2024,DePaulaEtAl2024,SarkarEtAl2025,SaleemEtAl2026}. These results show why the spacetime metric and the electromagnetic characteristic metric must be treated separately.

Our model is built from a positive mixture of simple power kernels. Positivity lets us prove the main inequalities term by term. Finite self-energy fixes the lower end of the exponent interval, while causal magnetic propagation fixes the upper end. We then choose the minimal two-kernel representative and obtain an exact black hole solution in terms of incomplete beta functions. To the best of our knowledge, the combined construction is new: it joins the positive power-mixture rule, the physically selected exponent interval, the exact two-kernel geometry, and a strict one-horizon result for $M_0\geq0$. The individual fractional-power kernels have related forms in the literature; the new elements are their constructive mixture, the resulting solution, and its global theorem.

The paper is arranged as follows. Section~\ref{sec:setup} gives the field equations and the characteristic and stability conditions. Sections~\ref{sec:model} and \ref{sec:solution} introduce the model and derive the exact solution. Sections~\ref{sec:properties}--\ref{sec:causal} study its physical conditions, horizons, and causal structure. Section~\ref{sec:thermo} discusses thermodynamics. We then examine photon motion, illustrative shadow bounds, sample trajectories, and ISCO-truncated disk images. Section~\ref{sec:discussion} summarizes the main results.
\section{Field equations and known conditions}\label{sec:setup}

We use $G=c=1$, metric signature $(-,+,+,+)$, and
\begin{equation}
 \cF=\frac14 F_{\mu\nu}F^{\mu\nu}.
 \label{eq:Fdef}
\end{equation}
The action is
\begin{equation}
 S=\frac{1}{16\pi}\int \dd^4x\sqrt{-g}\left[R-4\cL(\cF)\right].
 \label{eq:action}
\end{equation}
This normalization is common in the Einstein--NLED literature \cite{Plebanski1970,Bronnikov2001,FanWang2016,Sorokin2022,VerbinEtAl2025}. Varying Eq.~\eqref{eq:action} gives
\begin{align}
 G_{\mu\nu}&=2\left(\LF F_{\mu\lambda}F_{\nu}{}^{\lambda}-g_{\mu\nu}\cL\right),
 \label{eq:einstein}\\
 \nabla_{\mu}\left(\LF F^{\mu\nu}\right)&=0,
 \qquad \nabla_{[\mu}F_{\nu\rho]}=0,
 \label{eq:nled}
\end{align}
and
\begin{equation}
 T_{\mu\nu}=\frac{1}{4\pi}\left(\LF F_{\mu\lambda}F_{\nu}{}^{\lambda}-g_{\mu\nu}\cL\right).
 \label{eq:stress}
\end{equation}
Standard derivations can be found in Refs.~\cite{Plebanski1970,Sorokin2022}.

For a static magnetic monopole we use
\begin{align}
 \dd s^2&=-f(r)\dd t^2+\frac{\dd r^2}{f(r)}+r^2\dd\Omega^2,
 \qquad f(r)=1-\frac{2m(r)}{r},
 \label{eq:metric}\\
 F&=Q_m\sin\theta\,\dd\theta\wedge\dd\phi,
 \qquad \cF=\frac{Q_m^2}{2r^4}.
 \label{eq:magnetic}
\end{align}
The Bianchi identity fixes the magnetic charge. With this field, the NLED equation is satisfied, and the $t$-$t$ Einstein equation reduces to
\begin{equation}
 m'(r)=r^2\cL\!\left(\frac{Q_m^2}{2r^4}\right).
 \label{eq:mprime}
\end{equation}
This is the standard mass equation for static magnetic Einstein--NLED solutions \cite{Bronnikov2001,Ma2015,FanWang2016,Kruglov2017,Bronnikov2022,LeeMyung2026}.

Up to a conformal factor, the effective metric of the extraordinary electromagnetic mode can be written as \cite{Boillat1970,NovelloEtAl2000,ObukhovRubilar2002,Schellstede2016}
\begin{equation}
 g_{\rm eff}^{\mu\nu}=\LF g^{\mu\nu}-\LFF F^{\mu}{}_{\lambda}F^{\lambda\nu}.
 \label{eq:geff}
\end{equation}
For the magnetic field, the important factor is
\begin{equation}
 \Phi=\LF+2\cF\LFF,
 \qquad
 \kappaem=\frac{\Phi}{\LF}.
 \label{eq:kappa}
\end{equation}
In our convention, $\LF>0$ and $\Phi>0$ keep the magnetic characteristic well defined. The condition $0<\kappaem\leq1$ places the extraordinary cone inside, or exactly on, the spacetime cone \cite{ShabadUsov2011,GibbonsHerdeiro2001,Schellstede2016}.

Moreno and Sarbach obtained sufficient conditions for linear stability outside a magnetic black hole \cite{MorenoSarbach2003}. Nomura, Yoshida and Soda later extended this analysis to a general theory with two invariants \cite{NomuraYoshidaSoda2020}. For the one-invariant model used here, we use the following sufficient conditions:
\begin{equation}
 \cL>0,
 \qquad \LF>0,
 \qquad 0<f(r)\kappaem(r)<3
 \quad (r>r_+).
 \label{eq:stability}
\end{equation}
The value 3 comes from the lowest coupled multipole, $\ell=2$ \cite{NomuraYoshidaSoda2020}.

\section{The power-mixture NLED model}\label{sec:model}

We start from the kernel
\begin{equation}
 \cL_{\beta,\gamma}(\cF)=
 \frac{(1+\beta\cF)^{1-\gamma}-1}{\beta(1-\gamma)},
 \qquad \beta>0.
 \label{eq:kernel}
\end{equation}
We then take a positive normalized mixture,
\begin{equation}
 \cL_{\mu}(\cF)=\int
 \cL_{\beta,\gamma}(\cF)\,\dd\mu(\beta,\gamma),
 \qquad \int\dd\mu=1.
 \label{eq:mixture}
\end{equation}
where the measure has support in
\begin{equation}
 \boxed{\frac14<\gamma\leq\frac12.}
 \label{eq:range}
\end{equation}
The positive measure makes $\cL_{\mu}$ a Bernstein function of $\cF$ \cite{SchillingBernstein2012} and preserves the physical inequalities derived below.

Direct differentiation gives
\begin{align}
 \LF&=\int(1+\beta\cF)^{-\gamma}\,\dd\mu>0,
 \label{eq:LFmix}\\
 \LFF&=-\int\gamma\beta(1+\beta\cF)^{-\gamma-1}\,\dd\mu<0,
 \label{eq:LFFmix}\\
 \Phi&=\int
 \frac{1+(1-2\gamma)\beta\cF}
 {(1+\beta\cF)^{\gamma+1}}\,\dd\mu>0.
 \label{eq:Phimix}
\end{align}
Since $\LFF<0$, we also have $0<\Phi\leq\LF$ and therefore
\begin{equation}
 0<\kappaem\leq1.
 \label{eq:kappabound}
\end{equation}
These inequalities hold at every magnetic field strength.

The Maxwell limit follows from the normalization of the measure:
\begin{equation}
 \cL_{\mu}(\cF)=\cF-\frac{\cF^2}{2}
 \int\gamma\beta\,\dd\mu+O(\cF^3).
 \label{eq:weakgeneral}
\end{equation}
None of the theory parameters depends on the mass or charge of a particular solution. This avoids a known problem of some reverse-engineered regular models \cite{BokulicEtAl2024,BokulicEtAl2026}.

The lower limit in Eq.~\eqref{eq:range} comes from the self-energy. Near the center, $\cF\sim r^{-4}$, and one kernel behaves as $\cL\sim\cF^{1-\gamma}$. Therefore,
\begin{equation}
 \int_0 r^2\cL\,\dd r
 \sim\int_0 r^{-2+4\gamma}\,\dd r.
 \label{eq:selfcriterion}
\end{equation}
This integral is finite only when $\gamma>1/4$. The upper limit follows from Eq.~\eqref{eq:Phimix}: at large field, the numerator remains positive when $\gamma\leq1/2$.

The simplest nontrivial choice contains two kernels. We take one exponent inside the allowed interval and the other at its upper boundary:
\begin{equation}
 (\gamma_1,\gamma_2)=\left(\frac13,\frac12\right).
 \label{eq:gammas}
\end{equation}
The value $1/3$ is not unique; it is a simple rational representative inside the allowed interval. The value $1/2$ saturates the causal upper bound. This pair gives a minimal two-kernel model with one interior and one boundary contribution.
With $\beta>0$, $\sigma>0$, and $0\leq\xi\leq1$, the Lagrangian is
\begin{align}
 \cL(\cF)={}&
 \frac{3(1-\xi)}{2\beta}
 \left[(1+\beta\cF)^{2/3}-1\right]
 +\frac{2\xi}{\sigma\beta}
 \left[(1+\sigma\beta\cF)^{1/2}-1\right].
 \label{eq:modelL}
\end{align}
Its derivatives are
\begin{align}
 \LF={}&(1-\xi)(1+\beta\cF)^{-1/3}
 +\xi(1+\sigma\beta\cF)^{-1/2},
 \label{eq:modelLF}\\
 \LFF={}&-\frac{(1-\xi)\beta}{3(1+\beta\cF)^{4/3}}
 -\frac{\xi\sigma\beta}{2(1+\sigma\beta\cF)^{3/2}},
 \label{eq:modelLFF}\\
 \Phi={}&(1-\xi)\frac{1+\beta\cF/3}{(1+\beta\cF)^{4/3}}
 +\xi\frac{1}{(1+\sigma\beta\cF)^{3/2}}.
 \label{eq:modelPhi}
\end{align}
Figure~\ref{fig:constitutive} shows these functions.

At weak field,
\begin{align}
 \cL={}&\cF-\beta\left[\frac{1-\xi}{6}+\frac{\xi\sigma}{4}\right]\cF^2
 +\beta^2\left[\frac{2(1-\xi)}{27}+\frac{\xi\sigma^2}{8}\right]\cF^3
 +O(\cF^4).
 \label{eq:weakmodel}
\end{align}
The quadratic coefficient is negative because the model is concave in $\cF$. For this reason, it should not be identified with the one-loop Heisenberg--Euler action.
\begin{figure}[t]
\centering
\includegraphics[width=0.96\textwidth]{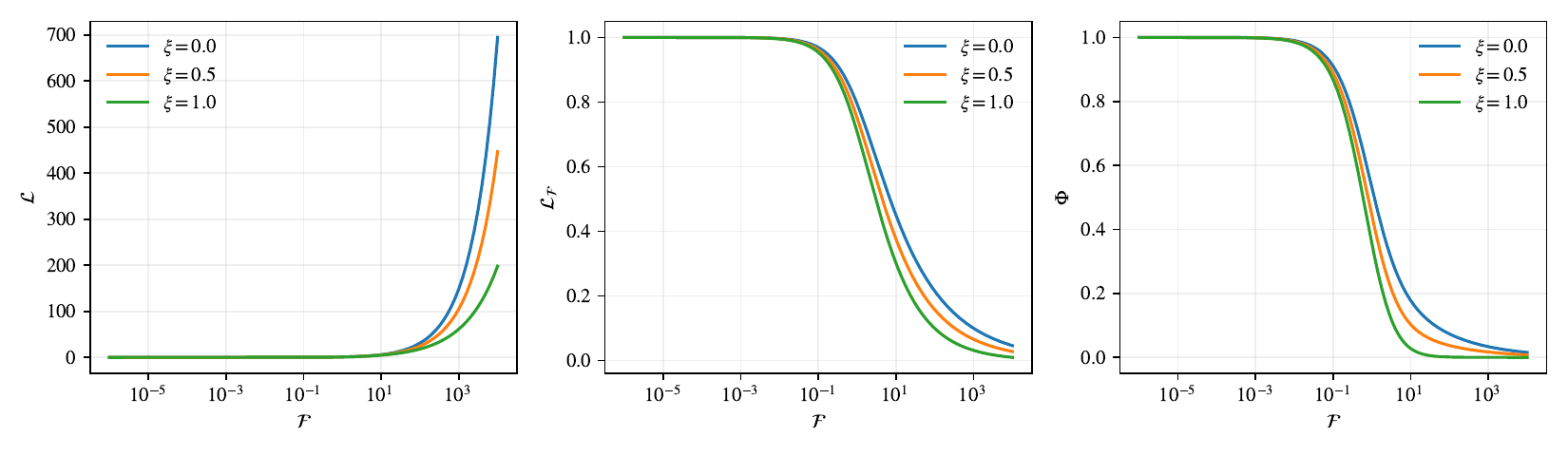}
\caption{Constitutive functions of the two-kernel model for $\sigma=1$. Every curve satisfies $\LF>0$, $\LFF<0$, and $\Phi>0$. Thus the extraordinary magnetic cone remains well defined for all field strengths.}
\label{fig:constitutive}
\end{figure}

\section{Exact static solution and finite self-energy}\label{sec:solution}

For one kernel, define
\begin{equation}
 b^4=\frac{\beta Q_m^2}{2},
 \qquad
 z(r)=\frac{r^4}{r^4+b^4}.
 \label{eq:bz}
\end{equation}
The radial integral in Eq.~\eqref{eq:mprime} is
\begin{align}
 \mathcal I_{\gamma}(r;b)
 ={}&\int_0^r x^2
 \left[\left(1+\frac{b^4}{x^4}\right)^{1-\gamma}-1\right]\dd x
 \nonumber\\
 ={}&\frac{b^3}{4}
 B_{z(r)}\!\left(\gamma-\frac14,-\frac34\right)
 -\frac{r^3}{3}.
 \label{eq:I}
\end{align}
Here $B_z(a,b)$ denotes the lower incomplete beta function. The result follows by using $z=x^4/(x^4+b^4)$ together with standard beta-function identities \cite{DLMF,GradshteynRyzhik2015}.

For the two-kernel model, set
\begin{equation}
 b_1^4=\frac{\beta Q_m^2}{2},
 \qquad
 b_2^4=\frac{\sigma\beta Q_m^2}{2}.
 \label{eq:b12}
\end{equation}
The exact mass function is
\begin{align}
 m(r)=M_0
 &+\frac{3(1-\xi)}{2\beta}\mathcal I_{1/3}(r;b_1)
 \nonumber\\
 &+\frac{2\xi}{\sigma\beta}\mathcal I_{1/2}(r;b_2),
 \label{eq:mass}
\end{align}
and
\begin{equation}
 \boxed{f(r)=1-\frac{2m(r)}{r}.}
 \label{eq:fexact}
\end{equation}
The constant $M_0$ is an independent Schwarzschild mass parameter. For the one-horizon analysis below, we restrict to $M_0\geq0$, including the electromagnetic-mass branch $M_0=0$. Differentiating Eq.~\eqref{eq:mass} directly recovers Eq.~\eqref{eq:mprime}.

The total electromagnetic self-energy is finite. For a single kernel, we find
\begin{equation}
 \int_0^\infty r^2
 \left[\left(1+\frac{b^4}{r^4}\right)^{1-\gamma}-1\right]\dd r
 =\frac{b^3(1-\gamma)}{3}
 B\!\left(\frac14,\gamma-\frac14\right).
 \label{eq:kernelenergy}
\end{equation}
We prove this result in Appendix~\ref{app:energy}. The two-kernel self-energy is
\begin{align}
 E_{\rm em}={}&\frac{(Q_m^2/2)^{3/4}}{3\beta^{1/4}}
 \Bigg[
 (1-\xi)B\!\left(\frac14,\frac1{12}\right)
 +\xi\sigma^{-1/4}B\!\left(\frac14,\frac14\right)
 \Bigg],
 \label{eq:Eem}
\end{align}
and the total mass is
\begin{equation}
 M=M_0+E_{\rm em}.
 \label{eq:ADM}
\end{equation}

At large radius, Eq.~\eqref{eq:weakmodel} leads to
\begin{align}
 f(r)={}&1-\frac{2M}{r}+\frac{Q_m^2}{r^2}
 -\frac{\beta Q_m^4}{10r^6}
 \left[\frac{1-\xi}{6}+\frac{\xi\sigma}{4}\right]
 +O(r^{-10}).
 \label{eq:asymptotic}
\end{align}
Thus, far from the center, the solution approaches the Reissner--Nordstr\"om form.

The center is singular. When $0\leq\xi<1$, the $\gamma=1/3$ term is dominant and
\begin{equation}
 m(r)=M_0+C_{1/3}r^{1/3}+o(r^{1/3}),
 \label{eq:centerM}
\end{equation}
where $C_{1/3}>0$. On the electromagnetic-mass branch, with $M_0=0$, the metric diverges as $r^{-2/3}$. This divergence is weaker than the Schwarzschild term $r^{-1}$, but the curvature is still singular. For $\xi=1$, one has $m(r)\sim C_{1/2}r$. The metric then approaches a global-monopole-type constant, while the curvature remains singular \cite{BarriolaVilenkin1989,RussoTownsend2026}.

\section{Energy conditions, causality, and exterior stability}\label{sec:properties}

For the magnetic stress tensor,
\begin{equation}
 \rho=\frac{\cL}{4\pi},
 \qquad p_r=-\rho,
 \qquad p_t=\frac{2\cF\LF-\cL}{4\pi}.
 \label{eq:pressures}
\end{equation}
These relations follow directly from Eq.~\eqref{eq:stress} and are standard in magnetic NLED \cite{Bronnikov2001,Bronnikov2022}.

The conditions $\cL\geq0$ and $\LF>0$ imply the weak energy condition. For the strong energy condition, the relevant combination satisfies
\begin{equation}
 \frac{\dd}{\dd\cF}\left(2\cF\LF-\cL\right)=\Phi>0,
 \label{eq:secproof}
\end{equation}
This combination vanishes at $\cF=0$. Since its derivative is positive, we have $2\cF\LF-\cL\geq0$.

For one kernel, with $y=1+\beta\cF$,
\begin{equation}
 \cL-\cF\LF=
 \frac{\gamma y^{1-\gamma}+(1-\gamma)y^{-\gamma}-1}
 {\beta(1-\gamma)}\geq0.
 \label{eq:decproof}
\end{equation}
The last inequality follows from the weighted arithmetic--geometric mean inequality. A positive mixture preserves it. Therefore, both the dominant and strong energy conditions hold. This is consistent with the general link between causality and energy conditions in NLED \cite{RussoTownsend2024}.

Equations~\eqref{eq:LFmix}--\eqref{eq:kappabound} also give
\begin{equation}
 \cL>0,
 \qquad \LF>0,
 \qquad 0<\kappaem\leq1.
 \label{eq:mainbounds}
\end{equation}
In the exterior of a positive-mass black hole, $0<f<1$. Therefore,
\begin{equation}
 0<f(r)\kappaem(r)<1<3,
 \qquad r>r_+.
 \label{eq:stabilityproof}
\end{equation}
Thus the sufficient stability conditions in Eq.~\eqref{eq:stability} hold everywhere outside the horizon. This result applies to the static magnetic exterior. It does not replace a direct calculation of the coupled quasinormal modes \cite{MorenoSarbach2003,NomuraYoshidaSoda2020,NomuraYoshida2022,DaghighEtAl2022,FathiGuzmanVillanueva2026,LiangEtAl2026}.

\section{One-horizon result}\label{sec:horizon}

\begin{lemma}\label{lem:onehorizon}
For the nontrivial magnetic solutions with $Q_m\neq0$ and $M_0\geq0$, the metric function satisfies $f'(r)>0$ for every $r>0$. Hence $f(r)=0$ has at most one positive root.
\end{lemma}

\noindent\textit{Proof.}
Define
\begin{equation}
 D(r)=m(r)-r m'(r).
 \label{eq:D}
\end{equation}
Using Eq.~\eqref{eq:mprime} and $r\cF'=-4\cF$, we find
\begin{align}
 f'(r)&=\frac{2D(r)}{r^2},
 \label{eq:fprime}\\
 D'(r)&=2r^2\left(2\cF\LF-\cL\right)\geq0.
 \label{eq:Dprime}
\end{align}
The second relation follows from Eq.~\eqref{eq:secproof}. For a nonzero magnetic field, $2\cF\LF-\cL>0$ at every finite $r>0$, so $D'(r)>0$. If $M_0>0$, then $D(0^+)=M_0>0$. If $M_0=0$, finite self-energy gives $D(0^+)=0$, and the strict increase of $D$ gives $D(r)>0$ for every $r>0$. Therefore,
\begin{equation}
 \boxed{f'(r)>0\qquad (r>0)}.
 \label{eq:monotonic}
\end{equation}
Thus $f(r)$ is strictly increasing and can have at most one positive root.\hfill$\square$

\medskip
We now determine when this root exists. For $M_0>0$, $f(r)\to-\infty$ as $r\to0^+$ and $f(r)\to1$ as $r\to\infty$, so there is exactly one horizon. The same conclusion holds for $M_0=0$ when $0\leq\xi<1$. For the pure boundary model, with $\xi=1$, define
\begin{equation}
 \lambda=\frac{|Q_m|}{\sqrt{\beta}}.
 \label{eq:lambda}
\end{equation}
Then
\begin{equation}
 f(0^+)=1-\frac{2\sqrt{2}\lambda}{\sqrt{\sigma}}.
 \label{eq:f0}
\end{equation}
The electromagnetic-mass branch is horizonless for
\begin{equation}
 \lambda\leq\frac{\sqrt{\sigma}}{2\sqrt{2}},
 \label{eq:threshold}
\end{equation}
and has one horizon above this value. At equality, the limiting zero is at the singular center, so it is not a regular extremal horizon. Therefore, within the $M_0\geq0$ sector, the solution has no inner Cauchy horizon and no extremal horizon at finite radius.

This result is consistent with the finite-self-energy argument of Ref.~\cite{HaleEtAl2026} and with other no-Cauchy-horizon results in nonlinear matter systems \cite{AnLiYang2021}. Here the stronger statement follows directly from the monotonicity of $f(r)$ in the sector considered above. Figure~\ref{fig:horizons} shows representative examples and the positive function $D$.

\begin{figure}[htp]
\centering
\includegraphics[width=0.96\textwidth]{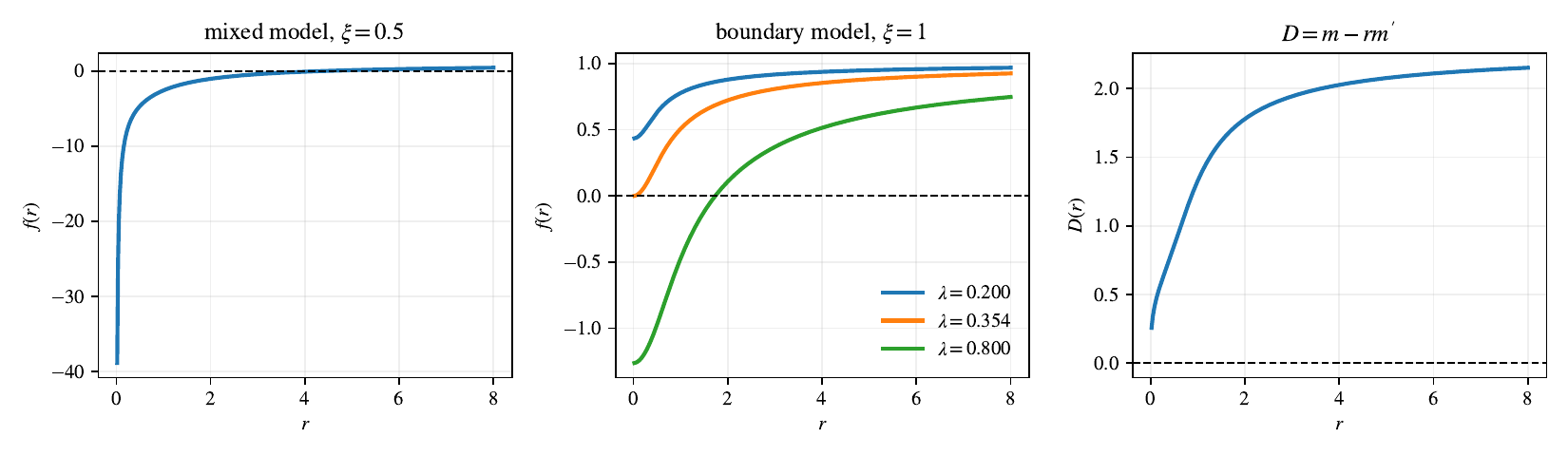}
\caption{Examples of the static solution in dimensionless variables. The first two panels show the metric function for the mixed and boundary models. The last panel shows $D=m-rm'$. Since $D>0$, we have $f'>0$, which proves the one-horizon result.}
\label{fig:horizons}
\end{figure}

\section{Maximal extension and causal structure}\label{sec:causal}

The one-horizon result makes the global analysis simple. We follow the usual Schwarzschild construction, while keeping the exact NLED function $f(r)$ at every step \cite{Eddington1924,Finkelstein1958,Kruskal1960,Szekeres1960,Penrose1965,HawkingEllis1973,Wald1984,FrolovNovikov1998}. The tortoise coordinate cannot be written in elementary form, but an exact coordinate construction is still possible.

\subsection{Null and Kruskal--Szekeres coordinates}

For the radial metric,
\begin{equation}
 \dd s_{(2)}^2=-f(r)\dd t^2+\frac{\dd r^2}{f(r)},
\end{equation}
we define
\begin{equation}
 r_*(r)=\int^r\frac{\dd s}{f(s)},
 \qquad u=t-r_*,\qquad v=t+r_*.
\end{equation}
Let $r_+$ be the unique horizon and
\begin{equation}
 \kappa=\frac12f'(r_+)>0.
\end{equation}
Near the horizon,
\begin{equation}
 f(r)=2\kappa(r-r_+)+O\!\bigl((r-r_+)^2\bigr).
\end{equation}
We separate the logarithmic term by writing
\begin{align}
 2\kappa r_*&=\ln|r-r_+|+H(r),\\
 H(r)&=\int_{r_+}^{r}\left[\frac{2\kappa}{f(s)}-\frac{1}{s-r_+}\right]\dd s.
\end{align}
The function $H(r)$ is regular at the horizon. We now introduce
\begin{equation}
 U=-e^{-\kappa u},\qquad V=e^{\kappa v}
\end{equation}
in the right exterior. The signs are then continued in the standard way to the other regions. If
\begin{equation}
 G(r)=(r-r_+)e^{H(r)},
\end{equation}
then
\begin{equation}
 -UV=G(r),
 \qquad
 \dd s_{(2)}^2=-\frac{f(r)}{\kappa^2G(r)}\dd U\dd V.
\end{equation}
Since
\begin{equation}
 \lim_{r\to r_+}\frac{f(r)}{G(r)}=2\kappa,
\end{equation}
the metric is regular at the horizon.

With
\begin{equation}
 T=\frac{V+U}{2},\qquad X=\frac{V-U}{2},
\end{equation}
constant-radius curves satisfy
\begin{equation}
 X^2-T^2=G(r).
\end{equation}
For one-horizon branches, $G(0)$ is finite and negative. We use the free normalization of $U$ and $V$ to set $G(0)=-1$. The singular boundaries are then
\begin{equation}
 T=\pm\sqrt{1+X^2}.
\end{equation}
These boundaries are spacelike. Every coordinate curve in our figures stops there, and no curve is continued through the curvature singularity.

\subsection{Three parameter choices}

To show how the parameters change the extension, we use three examples from the one-horizon branch. In the text they are denoted by $A_{\rm c}$--$C_{\rm c}$, while the shorter labels A--C are used inside the diagrams:
\begin{align}
 A_{\rm c}:&\quad (\xi,\sigma,Q_m,\beta)=(0.25,1,1,1),\\
 B_{\rm c}:&\quad (\xi,\sigma,Q_m,\beta)=(0.50,2,1,1),\\
 C_{\rm c}:&\quad (\xi,\sigma,Q_m,\beta)=(1,1,0.8,1),
\end{align}
with $M_0=0$. These examples are used only for the causal analysis and are different from the optical benchmark models introduced later in Table~\ref{tab:models}. Their horizon data are
\begin{center}
\begin{tabular}{c ccc}
\toprule
Model & $r_+$ & $\kappa$ & $\kappa r_+$\\
\midrule
$A_{\rm c}$ & 5.1684 & 0.0931 & 0.481\\
$B_{\rm c}$ & 4.0753 & 0.1153 & 0.470\\
$C_{\rm c}$ & 1.7354 & 0.2274 & 0.395\\
\bottomrule
\end{tabular}
\end{center}
The parameters change the horizon radius, the surface gravity, and the positions of the constant-radius curves. However, the global topology remains the same as long as the solution stays on the one-horizon branch.

Figure~\ref{fig:ks_three} shows the Kruskal--Szekeres extension. 
\begin{figure}[t]
\centering
\makebox[\textwidth][c]{\includegraphics[width=1.08\textwidth]{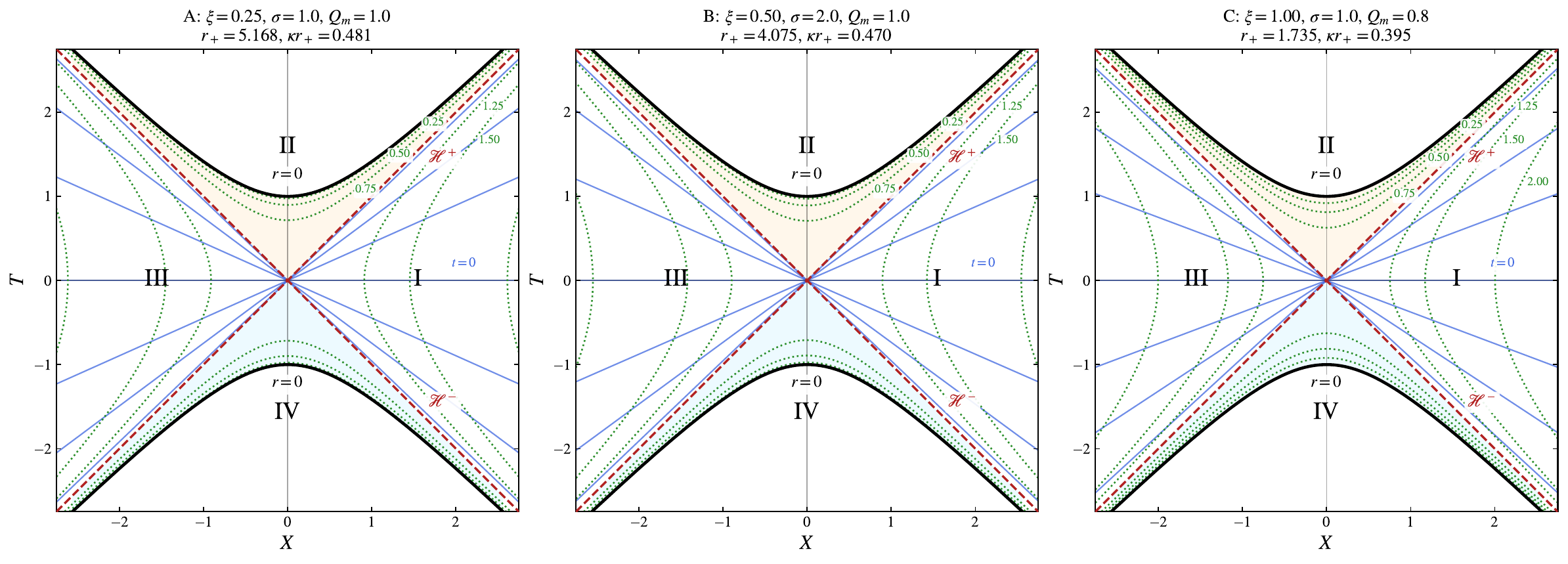}}
\caption{Kruskal--Szekeres extensions for the causal examples $A_{\rm c}$--$C_{\rm c}$, shown as A--C inside the panels. Green dotted curves have constant $r/r_+$, and blue curves have constant static time in regions I and III. Red dashed lines mark the horizons, while black curves mark the future and past spacelike singularities. The parameters move the constant-radius curves, but the three examples have the same one-horizon topology.}
\label{fig:ks_three}
\end{figure}
The constant-time curves are drawn only in regions I and III, where $t$ is timelike. The spacelike singularities are true boundaries of the spacetime, and every coordinate curve ends there.

\subsection{Standard conformal coordinates}

The standard compact coordinates are
\begin{equation}
 \widehat U=\arctan U,\qquad \widehat V=\arctan V,
\end{equation}
with
\begin{equation}
 \widehat T=\widehat V+\widehat U,
 \qquad
 \widehat X=\widehat V-\widehat U.
\end{equation}
The complete diagrams are shown in Fig.~\ref{fig:cp_standard}. 
\begin{figure}[t]
\centering
\makebox[\textwidth][c]{\includegraphics[width=1.08\textwidth]{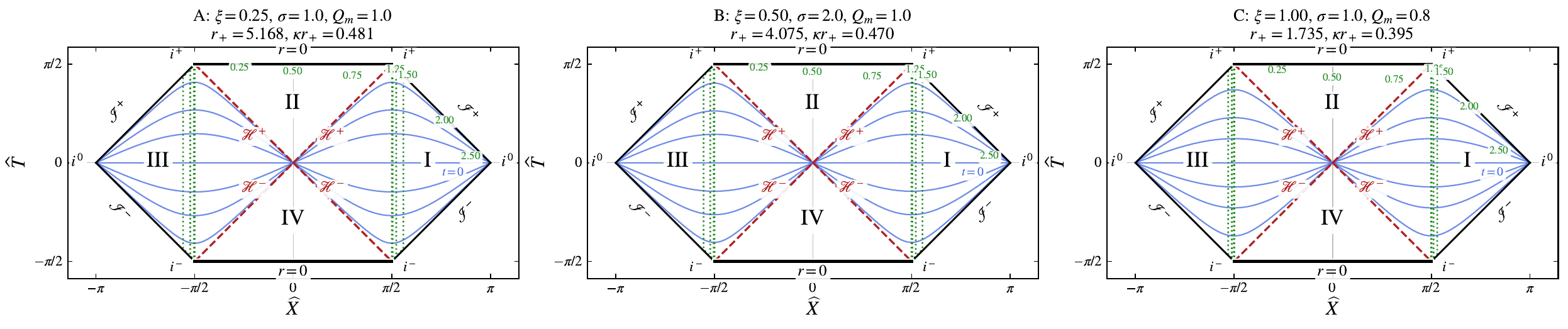}}
\caption{Standard Carter--Penrose diagrams for the same three causal examples. The full null and timelike infinities are included. Green dotted curves have constant $r/r_+$, and blue curves have constant static time in the two exterior regions. The singularities are spacelike and form the upper and lower boundaries.}
\label{fig:cp_standard}
\end{figure}
The asymptotic boundaries and constant-time curves terminate at the correct conformal infinities $i^0$, $i^+$, and $i^-$ in both exterior regions.

\subsection{Regular compact representation}

For a smoother representation, we also use
\begin{equation}
 \widetilde U=\arctan[\operatorname{arcsinh}(U)],
 \qquad
 \widetilde V=\arctan[\operatorname{arcsinh}(V)],
\end{equation}
followed by
\begin{equation}
 \widetilde T=\widetilde V+\widetilde U,
 \qquad
 \widetilde X=\widetilde V-\widetilde U.
\end{equation}
This map changes only the visual shape of the compact diagram and leaves the causal relations unchanged. Figure~\ref{fig:cp_regular} shows the corresponding representation. 
\begin{figure}[t]
\centering
\makebox[\textwidth][c]{\includegraphics[width=1.08\textwidth]{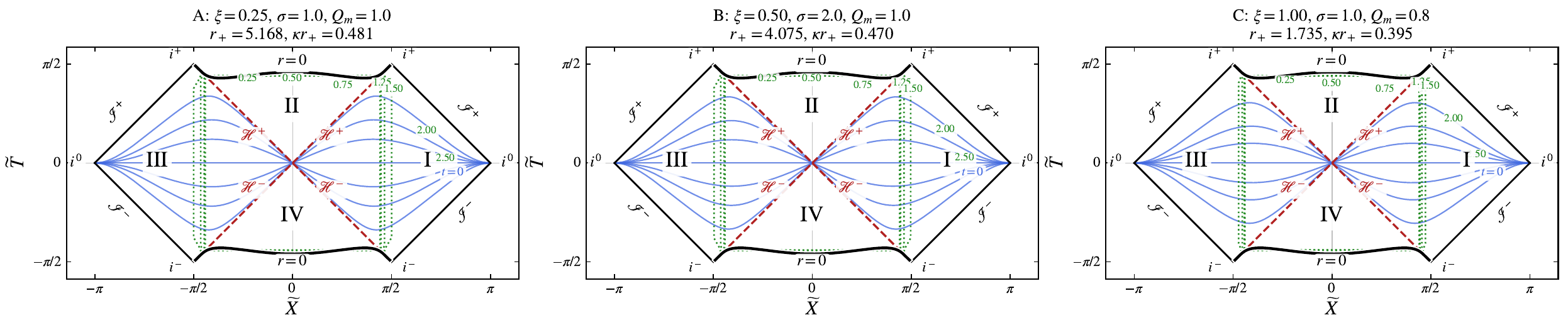}}
\caption{Regular compact Carter--Penrose representations. The causal structure is the same as in Fig.~\ref{fig:cp_standard}, but the spacelike singularities are drawn as smooth curves. All asymptotic and singular endpoints are included explicitly. No curve passes through a singular boundary.}
\label{fig:cp_regular}
\end{figure}
The singular boundaries and all asymptotic branches reach their proper conformal endpoints.

\subsection{Change of topology in the boundary model}

For $\xi=1$ and $M_0=0$, define
\begin{equation}
 \lambda=\frac{|Q_m|}{\sqrt{\beta}},
 \qquad
 \lambda_c=\frac{\sqrt{\sigma}}{2\sqrt{2}}.
\end{equation}
Near the center,
\begin{equation}
 f(r)=1-\frac{2\sqrt{2}\lambda}{\sqrt{\sigma}}
 +\frac{4r^2}{3\sigma\beta}+O(r^4).
\end{equation}
When $\lambda>\lambda_c$, the spacetime has one horizon and a spacelike singularity. When $\lambda<\lambda_c$, there is no horizon and the singularity is timelike and naked. At the boundary value $\lambda=\lambda_c$,
\begin{equation}
 f(r)=\frac{4r^2}{3\sigma\beta}+O(r^4),
 \qquad
 r_*=-\frac{3\sigma\beta}{4r}+O(r),
\end{equation}
the singularity becomes null. It remains a true curvature singularity and is not a regular extremal horizon.

\section{Thermodynamics}\label{sec:thermo}

\begin{figure}[!t]
\centering
\includegraphics[width=0.94\textwidth]{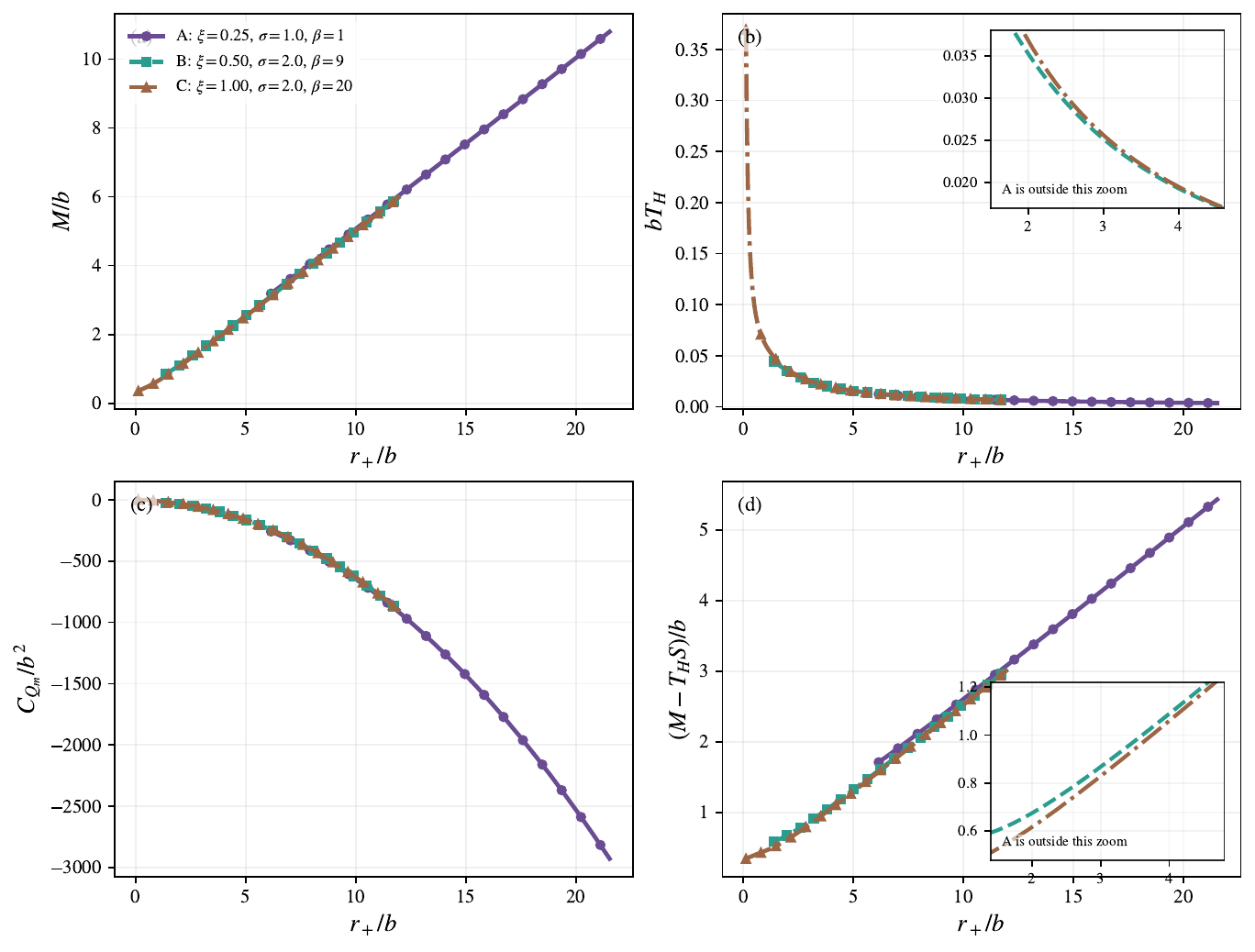}
\caption{Thermodynamic functions for three fixed-coupling families with $Q_m=1$. The scale is $b=(\beta Q_m^2/2)^{1/4}$. We use different line styles and markers because the full curves are close. The transparent insets show the small B--C separation in their common small-$r_+/b$ range; branch A starts outside these windows. For all three families, the temperature is positive and the fixed-charge heat capacity is negative on the branches shown.}
\label{fig:thermo}
\end{figure}

Let $r_+$ be the event-horizon radius. The horizon equation gives
\begin{equation}
 M_0=\frac{r_+}{2}-m_{\rm em}(r_+),
 \label{eq:M0h}
\end{equation}
where $m_{\rm em}=m-M_0$. The ADM mass can therefore be written as
\begin{equation}
 M(r_+,Q_m,\beta,\sigma,\xi)
 =\frac{r_+}{2}+\int_{r_+}^{\infty}r^2\cL(\cF)\,\dd r.
 \label{eq:Mh}
\end{equation}
The temperature and entropy are
\begin{equation}
 T_H=\frac{1-2r_+^2\cL(\cF_+)}{4\pi r_+},
 \qquad S=\pi r_+^2,
 \label{eq:TS}
\end{equation}
where $\cF_+=Q_m^2/(2r_+^4)$. These are the standard surface-gravity and area-law expressions of Einstein gravity \cite{Bekenstein1973,Hawking1975,Wald1993}. Together with the near-center expansion of the pure boundary branch, these expressions show that $T_H$ remains positive for every horizon at finite radius and vanishes continuously as $\lambda\to\lambda_c^{+}$.

At fixed charge and fixed couplings,
\begin{equation}
 \left(\frac{\partial M}{\partial r_+}\right)_{Q_m,\beta,\sigma,\xi}
 =\frac12-r_+^2\cL(\cF_+)
 =T_H\frac{\dd S}{\dd r_+}.
 \label{eq:firstlawcheck}
\end{equation}
When all parameters are varied, the first law takes the form
\begin{equation}
 \dd M=T_H\dd S+\Psi_m\dd Q_m
 +\Pi_{\beta}\dd\beta+\Pi_{\sigma}\dd\sigma+\Pi_{\xi}\dd\xi,
 \label{eq:firstlaw}
\end{equation}
with
\begin{align}
 \Psi_m&=Q_m\int_{r_+}^{\infty}\frac{\LF}{r^2}\,\dd r,
 \label{eq:potential}\\
 \Pi_a&=\int_{r_+}^{\infty}r^2\frac{\partial\cL}{\partial a}\,\dd r,
 \qquad a\in\{\beta,\sigma,\xi\}.
 \label{eq:couplingpotentials}
\end{align}
First laws with variable NLED couplings have been discussed in Refs.~\cite{Smarr1973,Breton2003,FernandoKrug2003,CaiPangWang2004,Fernando2006,ZhangGao2018,BokulicThermo2021,BalartFernando2021,CroneyEtAl2025,AbeEtAl2026}. The parameter $\beta$ has dimensions of length squared, while $\sigma$ and $\xi$ are dimensionless. Euler scaling then gives
\begin{equation}
 M=2T_HS+\Psi_m Q_m+2\beta\Pi_{\beta}.
 \label{eq:Smarr}
\end{equation}

The heat capacity at fixed charge can be written directly, without numerical differentiation:
\begin{equation}
 C_{Q_m}=\frac{2\pi r_+^2\left[1-2r_+^2\cL(\cF_+)\right]}
 {-1+8r_+^2\cF_+\LF(\cF_+)-2r_+^2\cL(\cF_+)}.
 \label{eq:CQ}
\end{equation}
To compare different nonlinear scales, we use
\begin{equation}
 b=\left(\frac{\beta Q_m^2}{2}\right)^{1/4}.
\end{equation}
Figure~\ref{fig:thermo} shows three representative fixed-coupling families. The temperature is positive and $C_{Q_m}$ is negative on all displayed branches. The insets resolve the small separation between the more nonlinear B and C families; branch A starts outside the zoomed ranges. These asymptotically flat branches are therefore locally unstable in the fixed-charge canonical ensemble. A statement about global thermodynamic stability would require a specified cavity or another external environment.

\section{Ordinary and extraordinary photon motion}\label{sec:photons}

In the geometric-optics limit, an $\cL(\cF)$ theory on a magnetic background has an ordinary branch and an extraordinary optical branch \cite{Boillat1970,NovelloEtAl2000,ObukhovRubilar2002,Schellstede2016,DePaulaEtAl2023,DePaulaEtAl2026,Escobar2026}. Similar effective-metric effects have been used in recent NLED shadow and imaging studies \cite{OkyayOvgun2022,UniyalPantigOvgun2023,UniyalEtAl2023,CimdikerEtAl2026}. After removing conformal factors, which do not change null trajectories, their line elements are
\begin{align}
 \dd s_{\rm ord}^2&=-f\dd t^2+\frac{\dd r^2}{f}+r^2\dd\Omega^2,
 \label{eq:ordmetric}\\
 \dd s_{\rm ext}^2&=-f\dd t^2+\frac{\dd r^2}{f}
 +\frac{r^2}{\kappaem(r)}\dd\Omega^2.
 \label{eq:extmetric}
\end{align}
The second metric follows from Eq.~\eqref{eq:geff}. Since $0<\kappaem\leq1$, its angular sector remains regular everywhere in the magnetic exterior.

It is useful to define
\begin{equation}
 C_{\rm ord}(r)=r^2,
 \qquad C_{\rm ext}(r)=\frac{r^2}{\kappaem(r)}.
 \label{eq:Cmodes}
\end{equation}
For equatorial null motion, the radial equation is
\begin{equation}
 \dot r^2=E^2-L^2\frac{f(r)}{C_s(r)},
 \qquad s\in\{\mathrm{ord},\mathrm{ext}\}.
 \label{eq:nullradial}
\end{equation}
An unstable circular photon orbit satisfies
\begin{equation}
 \left.\frac{\dd}{\dd r}\left(\frac{f}{C_s}\right)\right|_{r=r_{\rm ph}^{s}}=0,
 \label{eq:phcondition}
\end{equation}
and its critical impact parameter is
\begin{equation}
 (b_{\rm ph}^{s})^2=\left.\frac{C_s}{f}\right|_{r=r_{\rm ph}^{s}}.
 \label{eq:bcrit}
\end{equation}
For a distant observer, the geometrical shadow diameter is $2b_{\rm ph}^{s}$. The instability rate measured in coordinate time is
\begin{equation}
 (\lambda_s)^2=-\frac{f(r_{\rm ph}^{s})C_s(r_{\rm ph}^{s})}{2}
 \left.\frac{\dd^2}{\dd r^2}\left(\frac{f}{C_s}\right)\right|_{r=r_{\rm ph}^{s}}.
 \label{eq:lyapunov}
\end{equation}

Neutral disk matter follows the background metric. Its angular velocity, energy, and angular momentum on a circular orbit are
\begin{align}
 \Omega^2&=\frac{f'}{2r},
 \label{eq:Omega}\\
 E_c^2&=\frac{2f^2}{2f-rf'},
 \qquad
 L_c^2=\frac{r^3f'}{2f-rf'}.
 \label{eq:ELcirc}
\end{align}
The innermost stable circular orbit (ISCO) obeys
\begin{equation}
 rff''-2r(f')^2+3ff'=0.
 \label{eq:isco}
\end{equation}
These are standard relations for static spherical metrics. They also determine the inner edge used in thin disk models \cite{PageThorne1974,Luminet1979,UniyalPantigOvgun2023,UniyalEtAl2024}.

For the optical calculations, we set the ADM mass to $M=1$ and use the three parameter sets in Table~\ref{tab:models}; all have $M_0>0$. Model A is close to the Maxwell regime, while B and C have stronger nonlinear responses. The corresponding horizon, photon, shadow, and ISCO quantities are listed in Table~\ref{tab:optical}.

\begin{table}[htp]
\centering
\caption{Representative models used in the optical calculations. The ADM mass is $M=1$.}
\label{tab:models}
\begin{tabular}{ccccc}
\toprule
Model & $Q_m/M$ & $\beta/M^2$ & $\sigma$ & $\xi$\\
\midrule
A & 0.35 & 1  & 1 & 0.25\\
B & 0.60 & 9  & 2 & 0.50\\
C & 0.80 & 16 & 1 & 1.00\\
\bottomrule
\end{tabular}
\end{table}

\begin{table}[htp]
\centering
\caption{Horizon, photon, shadow, and ISCO quantities for the models in Table~\ref{tab:models}.}
\label{tab:optical}
\begin{tabular}{cccccccc}
\toprule
Model & $M_0/M$ & $r_+/M$ & $r_{\rm ph}^{\rm ord}/M$ & $r_{\rm ph}^{\rm ext}/M$ & $2b_{\rm ph}^{\rm ord}/M$ & $2b_{\rm ph}^{\rm ext}/M$ & $r_{\rm ISCO}/M$\\
\midrule
A & 0.445 & 1.937 & 2.916 & 2.917 & 10.175 & 10.178 & 5.812\\
B & 0.420 & 1.802 & 2.738 & 2.805 & 9.719 & 9.892 & 5.420\\
C & 0.474 & 1.617 & 2.493 & 2.694 & 9.102 & 9.603 & 4.895\\
\bottomrule
\end{tabular}
\end{table}

\begin{figure}[htp]
\centering
\includegraphics[width=0.94\textwidth]{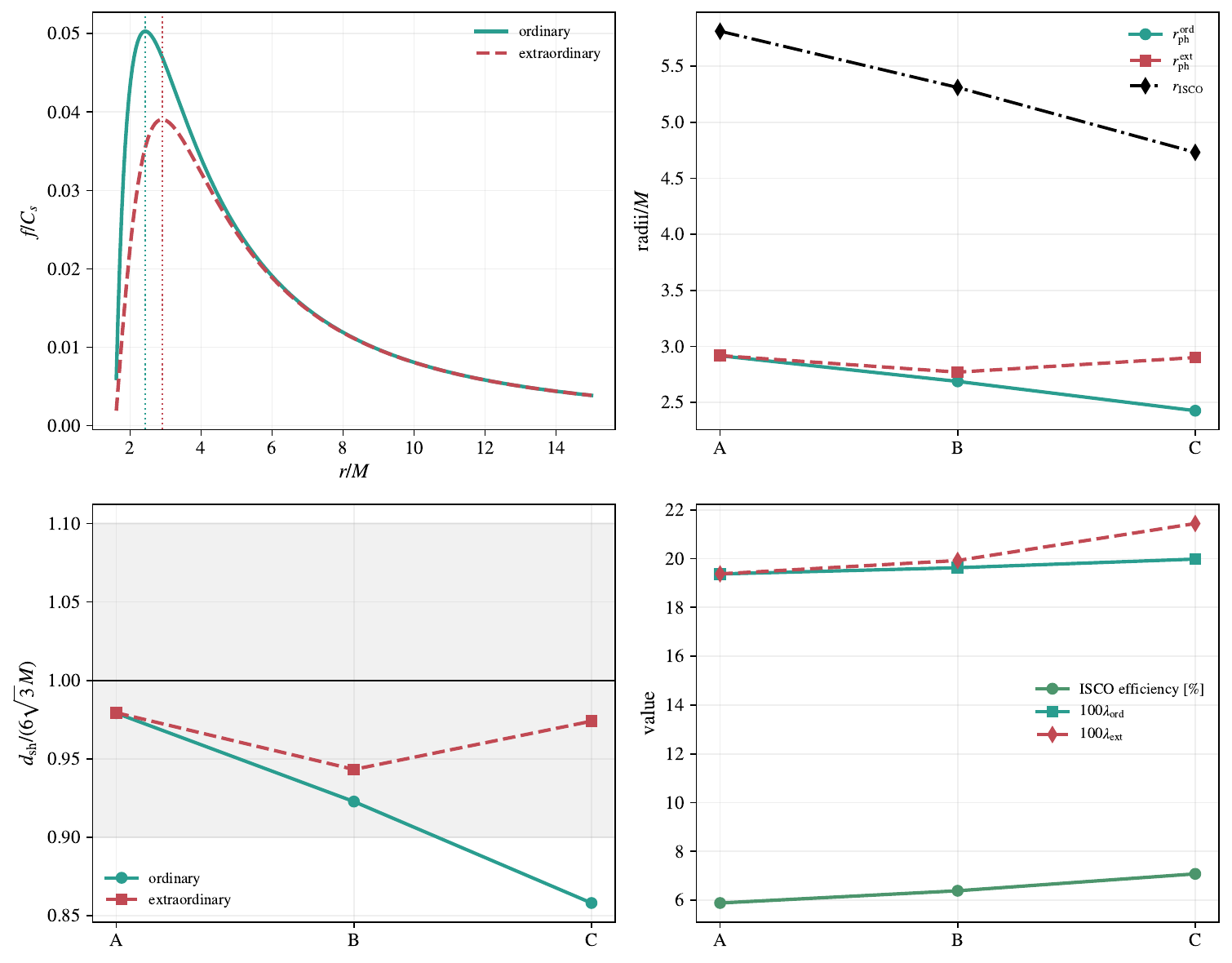}
\caption{Photon and circular-orbit observables. The upper-left panel shows the ordinary and extraordinary photon potentials for model C. The remaining panels compare the characteristic radii, shadow diameters, ISCO efficiency, and photon-orbit instability rates for models A--C. The shaded region is the simple $10\%$ shadow-size band used below.}
\label{fig:photonmotion}
\end{figure}

The ordinary shadow is smaller than the Schwarzschild diameter $6\sqrt{3}M$ by about $2.1\%$, $6.5\%$, and $12.4\%$ for models A, B, and C. The extraordinary branch changes the corresponding diameter by less than $0.1\%$, about $1.8\%$, and about $5.5\%$. Thus the optical correction is negligible in model A but becomes comparable to the geometric charge effect in model C. Over the same sequence, the background ISCO moves from $5.812M$ to $4.895M$, the extraordinary photon radius separates further from the ordinary one, and the ISCO efficiency rises modestly. Figure~\ref{fig:photonmotion} also shows a growing separation between the photon-orbit instability rates.

The outward shift follows from the effective angular function: because $0<\kappaem\leq1$, one has $C_{\rm ext}=r^2/\kappaem\geq C_{\rm ord}$ at the same radius, while the circular-orbit condition also changes the orbit location. Both Lyapunov exponents are positive, so all listed photon orbits are unstable. These rates describe the local instability of null rays and do not replace a coupled gravitational--electromagnetic perturbation analysis.

\section{Simple shadow bounds}\label{sec:shadow}

\begin{figure}[!t]
\centering
\includegraphics[width=0.92\textwidth]{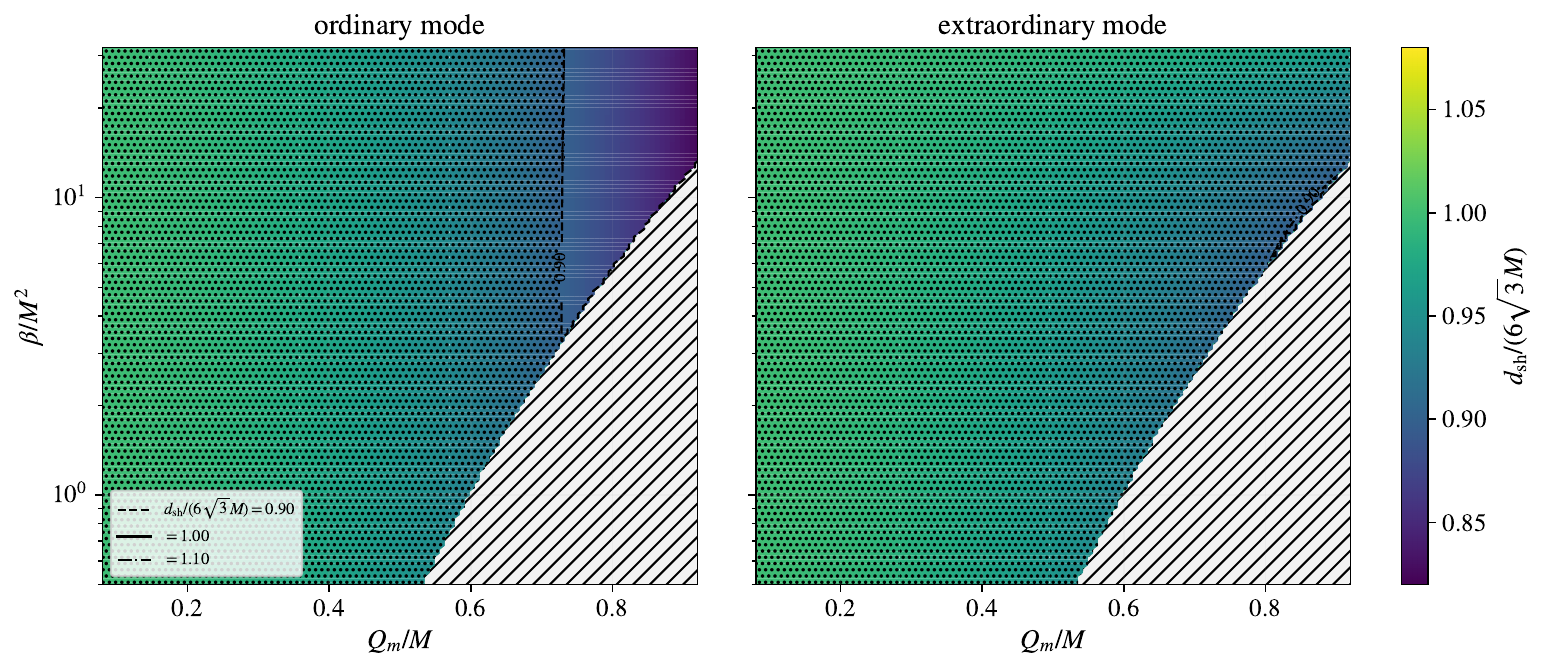}
\caption{Simple shadow-size scan for $\xi=0.5$ and $\sigma=2$. The color scale gives $d_{\rm sh}/(6\sqrt{3}M)$. The solid black contour shows the Schwarzschild value, and the labelled outer contours show the simple $10\%$ band of Eq.~\eqref{eq:ehtband}. The hatched lower-right region has $M_0<0$ and is excluded. The dotted overlay marks the part of the plane inside the illustrative band. In the ordinary panel, the nearly vertical $0.90$ contour shows a rapid but continuous exit from the band, not a discontinuity.}
\label{fig:shadowconstraints}
\end{figure}

An NLED shadow must be calculated from the characteristic geometry, not only from the background metric \cite{DePaulaEtAl2023,UniyalEtAl2023,DePaulaEtAl2026}. Recent optical studies have also used EHT data to constrain several NLED black hole families \cite{UniyalPantigOvgun2023,UniyalEtAl2024,RazaEtAl2024,LambiaseEtAl2025,CimdikerEtAl2026,UktamovEtAl2026}. When two optical branches are present, the result can also depend on polarization \cite{DePaulaEtAl2026}. For M87*, the EHT measured a ring diameter of $42\pm3\,\mu$as \cite{EHTM87VI2019}. For Sgr~A*, the calibrated image size was found to be within about $10\%$ of the Kerr prediction \cite{EHTSgrAVI2022}. The observed ring and the mathematical shadow are not the same object. We therefore use the Sgr~A* result only to define a simple illustrative band:
\begin{equation}
 0.9\leq\frac{d_{\rm sh}}{6\sqrt{3}M}\leq1.1.
 \label{eq:ehtband}
\end{equation}
This band is illustrative and is not a likelihood analysis.

Figure~\ref{fig:shadowconstraints} shows the scan for $\xi=0.5$, $\sigma=2$, and $M=1$. We exclude the hatched region because it has $M_0<0$. The ordinary shadow becomes smaller mainly as $Q_m/M$ increases. The extraordinary cone increases the critical impact parameter, especially at large $\beta/M^2$. It therefore partly compensates the geometrical reduction of the shadow. This effect is small at weak coupling, but it becomes visible in the stronger-coupling part of the plane. In the ordinary panel, the nearly vertical contour near $Q_m/M\simeq0.73$ is the lower edge, $d_{\rm sh}/(6\sqrt{3}M)=0.90$, of the illustrative band. It is not a discontinuity. The shadow size changes rapidly but continuously there and has only a weak dependence on $\beta/M^2$ over the plotted range.

The scan shows that the ordinary branch leaves the $10\%$ band as the magnetic charge grows. On the adopted grid, ordinary rays remain inside the band up to about $Q_m/M\simeq0.72$, while extraordinary rays still have allowed points up to about $Q_m/M\simeq0.92$. These grid-dependent values are not observational constraints; they only show that the effective optical branch matters in comparisons with EHT size information.

\section{Sample photon trajectories}\label{sec:rays}

\begin{figure}[!t]
\centering
\includegraphics[width=0.97\textwidth]{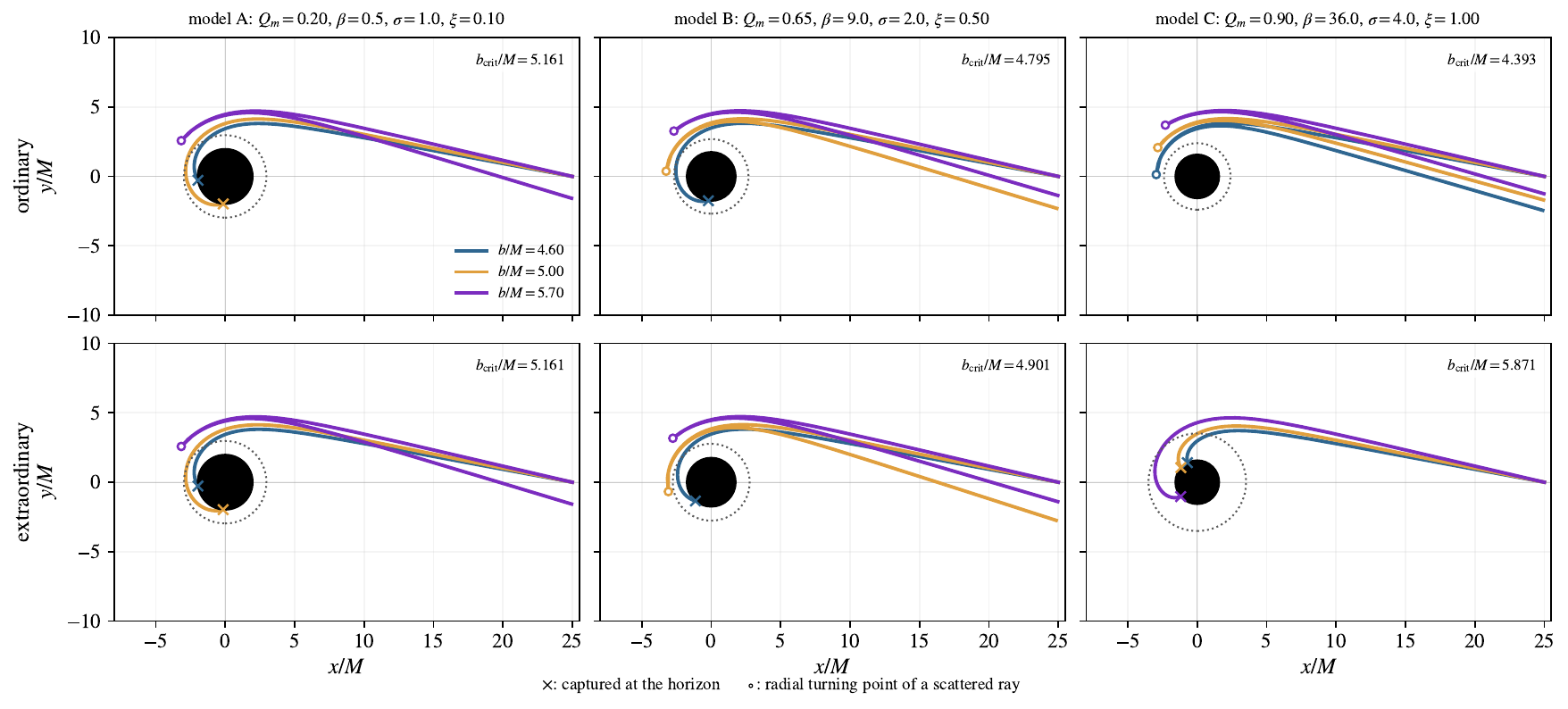}
\caption{Equatorial photon trajectories for the diagnostic sets $A_{\rm tr}$--$C_{\rm tr}$. The upper and lower rows show ordinary and extraordinary rays. Every panel uses the same impact parameters, $b/M=4.60$, $5.00$, and $5.70$. A cross marks a ray captured by the horizon, while an open circle marks the turning point of a scattered ray. The black disk is the event horizon, and the dotted circle shows the unstable photon-orbit radius.}
\label{fig:rays}
\end{figure}

Before turning to disk images, we first look at a few individual null rays. If every panel were normalized by its own critical impact parameter, some differences would be hidden. We therefore use the same values, $b/M=4.60$, $5.00$, and $5.70$, in every panel of Fig.~\ref{fig:rays}. For this diagnostic, the three parameter sets are
\begin{align}
 A_{\rm tr}&:(Q_m,\beta,\sigma,\xi)=(0.20,0.5,1,0.10),\\
 B_{\rm tr}&:(Q_m,\beta,\sigma,\xi)=(0.65,9,2,0.50),\\
 C_{\rm tr}&:(Q_m,\beta,\sigma,\xi)=(0.90,36,4,1.00).
\end{align}
These diagnostic sets are more widely separated than the optical benchmarks in Table~\ref{tab:models}. The columns show the three sets, and the rows show the ordinary and extraordinary branches. Because the critical impact parameter depends on both the model and the optical branch, the same incoming ray can be captured in one panel and scattered in another.

The weakly nonlinear set $A_{\rm tr}$ gives almost the same trajectories in both branches. The separation grows for $B_{\rm tr}$ and is largest for $C_{\rm tr}$, where the extraordinary critical impact parameter shifts strongly. The strongest set is used only to display this effect and is not an observational best fit.

\section{ISCO-truncated synthetic images and branch comparison}\label{sec:images}

Thin-disk and optically thin images are widely used as controlled probes of strong-field geometry \cite{Luminet1979,PageThorne1974,GrallaEtAl2019,UniyalPantigOvgun2023,UniyalEtAl2023,UniyalEtAl2024,UktamovEtAl2026}. We use a simple proxy model to isolate the difference between the two NLED photon branches without introducing detailed plasma physics.

We construct images of a geometrically thin equatorial disk. The disk matter is neutral and follows circular orbits of the background metric. The inner edge is placed at $r_{\rm ISCO}$, the outer edge at $40M$, and the emitted bolometric intensity is
\begin{equation}
 I_{\rm em}(r)\propto\left(\frac{r_{\rm ISCO}}{r}\right)^2
 \exp\!\left[-\frac{r-r_{\rm ISCO}}{30M}\right].
 \label{eq:emissivity}
\end{equation}
The images are produced with a ray-plane method implemented in custom Numba-accelerated Python code. We launch photons from an observer at $r_o=100M$ and propagate them in either the ordinary or extraordinary optical geometry. Each time a ray crosses the disk, its contribution is added to the intensity. Multiple crossings are kept with decreasing weights. Along each ray, Liouville's theorem gives
\begin{equation}
 I_{\rm obs}=g^4 I_{\rm em},
 \label{eq:intensity}
\end{equation}
where
\begin{equation}
 g=\frac{1}{u^t(1-\Omega\lambda)},
 \qquad
 u^t=\frac{1}{\sqrt{f-r^2\Omega^2}}.
 \label{eq:redshift}
\end{equation}
Here $\lambda=p_\phi/E$ is the conserved azimuthal impact parameter. The disk follows the background metric, while photons follow the optical metric of their branch. These are transparent thin-disk proxy images, not full general-relativistic radiative-transfer calculations \cite{Luminet1979,PageThorne1974,GrallaEtAl2019}. We use them to isolate the geometrical effect of the NLED characteristic cone.

We use the three models of Table~\ref{tab:models} and the inclinations $\theta_o=30^\circ$, $60^\circ$, and $80^\circ$. All panels cover the same screen range, $-22.5\leq X/M,Y/M\leq22.5$. For each model and inclination, the ordinary and extraordinary maps share one display scale, while all numerical comparisons are calculated from the raw intensity arrays. This joint normalization prevents small branch-dependent brightness changes from being hidden.

Figure~\ref{fig:images_ord} shows the ordinary branch. The image is nearly circular at $30^\circ$, while the Doppler-bright side and lensed lower arc become clearer at $60^\circ$. At $80^\circ$, the disk is almost edge-on and the ISCO-truncated inner edge is more visible. Inclination changes the morphology much more strongly than the differences among the three models. The dotted circle marks the calculated ordinary critical curve; no artificial intensity ring is added.
\begin{figure}[tbp]
\centering
\includegraphics[width=0.98\textwidth]{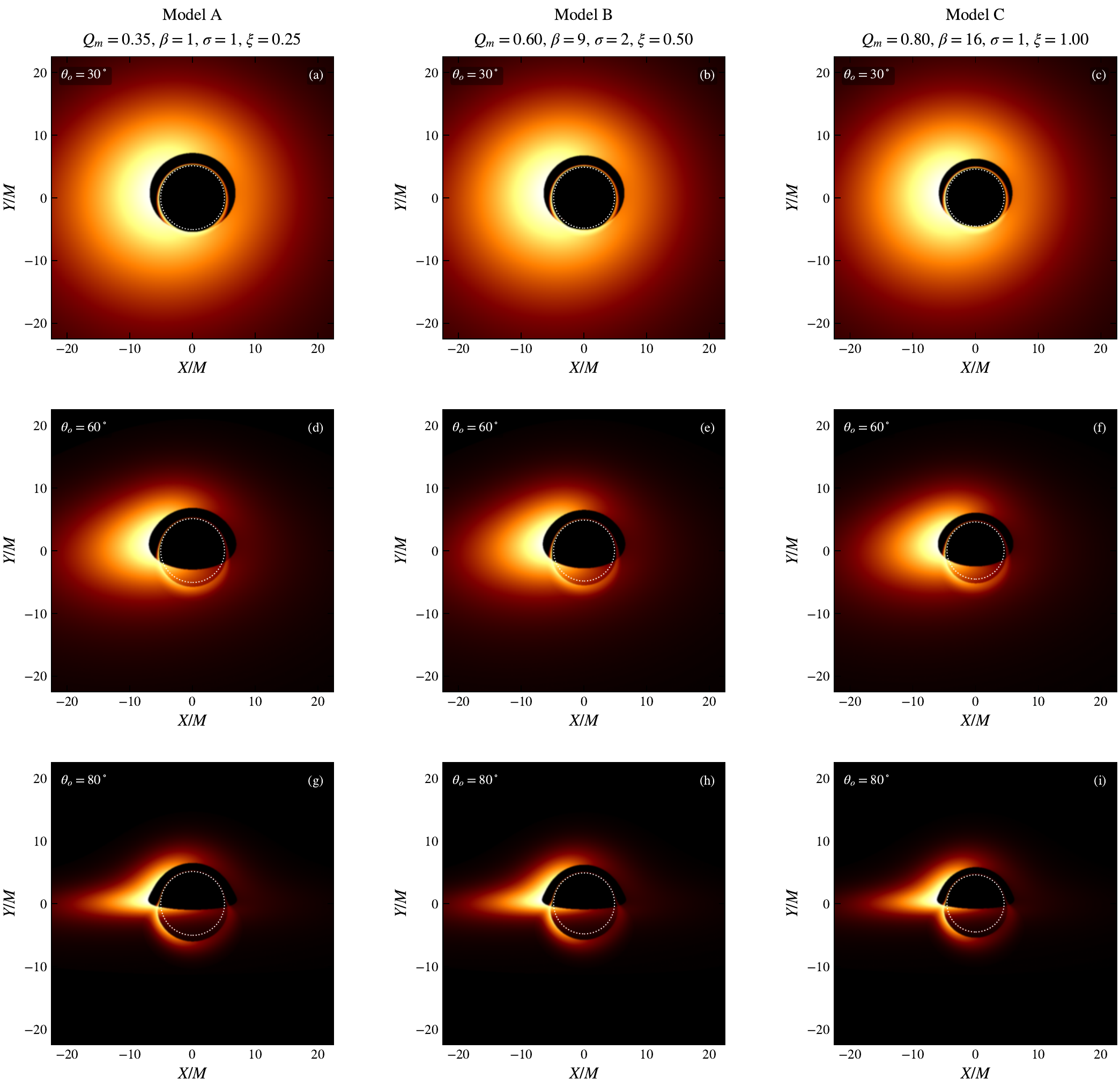}
\caption{ISCO-truncated thin-disk images for the ordinary branch. Columns show models A--C, and rows show $\theta_o=30^\circ$, $60^\circ$, and $80^\circ$. All panels use the same screen coordinates. For every fixed model and inclination, this image and its extraordinary partner use the same display scale. The dotted white circle marks the ordinary critical curve.}
\label{fig:images_ord}
\end{figure}

The full ordinary and extraordinary images are very similar, especially for model A. We therefore compare the branches through the residual map
\begin{equation}
 {\cal R}(X,Y)=
 \frac{I_{\rm ext}(X,Y)-I_{\rm ord}(X,Y)}
 {I_{\max}^{\rm joint}},
 \qquad
 I_{\max}^{\rm joint}
 =\max_{\substack{X,Y\\ s\in\{{\rm ord},{\rm ext}\}}} I_s(X,Y),
 \label{eq:image_residual}
\end{equation}
where the maximum is taken over both raw maps of the same model and inclination. Positive residuals mean that the extraordinary image is brighter at that screen position, while negative residuals mean that the ordinary image is brighter.

We also use two integrated quantities,
\begin{equation}
 \Delta_I=
 \frac{\int |I_{\rm ext}-I_{\rm ord}|\,\dd X\,\dd Y}
 {\int I_{\rm ord}\,\dd X\,\dd Y},
 \qquad
 R_F=
 \frac{\int I_{\rm ext}\,\dd X\,\dd Y}
 {\int I_{\rm ord}\,\dd X\,\dd Y}.
 \label{eq:image_metrics}
\end{equation}
The first quantity measures the total absolute image difference, while the second compares the total flux. Figure~\ref{fig:images_res} shows that model A is almost unchanged, model B develops a small coherent residual near the lensed inner edge, and model C gives the largest difference. This trend agrees with the critical-curve shifts in Table~\ref{tab:optical}. The residual is localized near the bright inner structures and the critical region rather than forming a new global morphology.
\begin{figure}[tbp]
\centering
\includegraphics[width=0.98\textwidth]{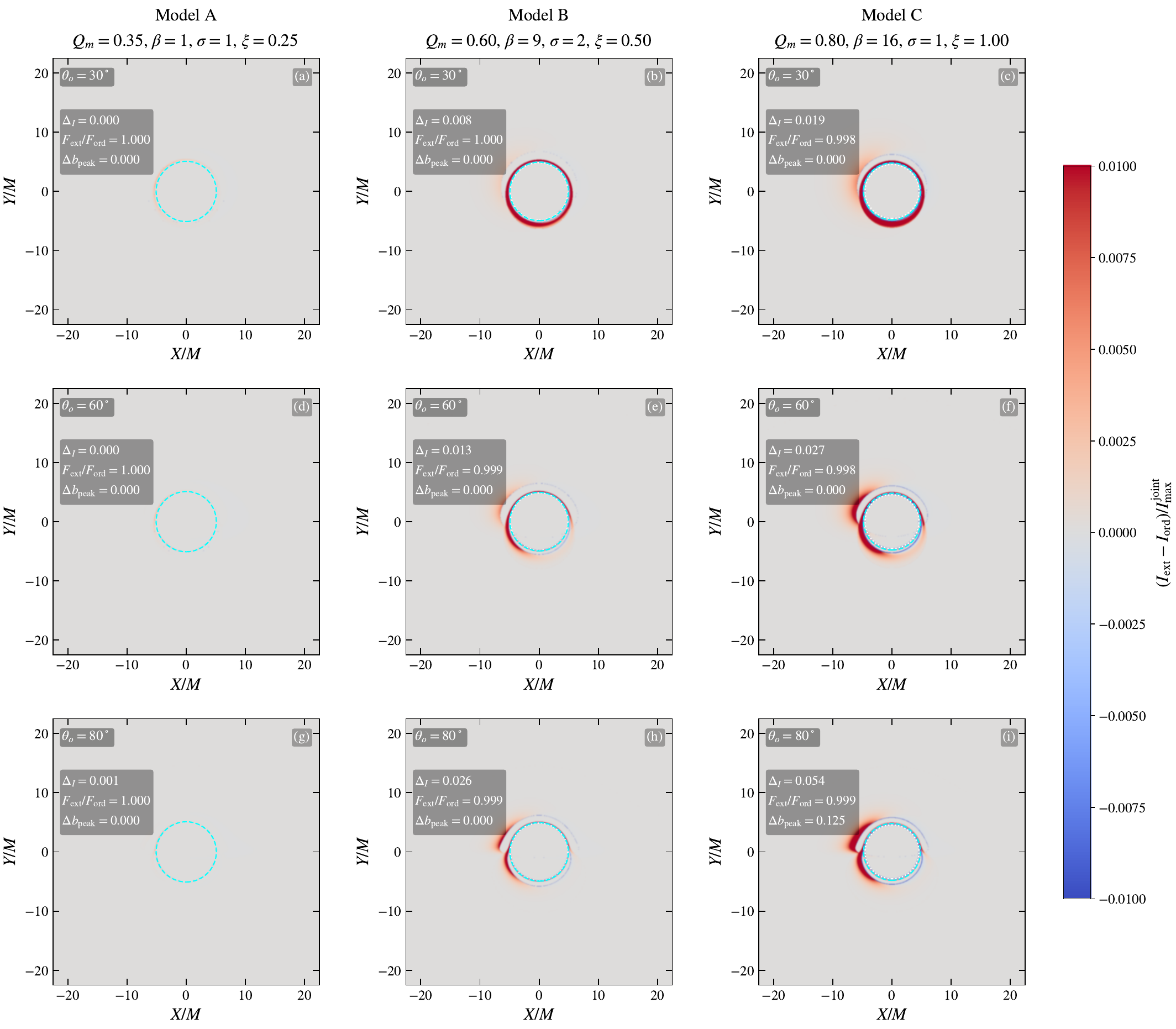}
\caption{Residual maps between the extraordinary and ordinary raw images, defined by Eq.~\eqref{eq:image_residual}. Red regions are brighter in the extraordinary image and blue regions are brighter in the ordinary image. White dotted and cyan dashed circles mark the ordinary and extraordinary critical curves. Each panel also reports $\Delta_I$, $R_F$, and the radial peak displacement obtained from the same raw arrays. A common residual color scale is used for all panels.}
\label{fig:images_res}
\end{figure}

We also calculate the azimuthally averaged profile
\begin{equation}
 \overline I_s(b)=\frac{1}{2\pi}\int_0^{2\pi} I_s(b,\varphi)\,\dd\varphi,
 \qquad s\in\{{\rm ord},{\rm ext}\},
 \label{eq:radial_profile}
\end{equation}
and define
\begin{equation}
 \Delta b_{\rm peak}=b_{\rm peak}^{\rm ext}-b_{\rm peak}^{\rm ord}.
 \label{eq:peak_shift}
\end{equation}
The main profiles in Fig.~\ref{fig:images_prof} nearly overlap. The transparent insets magnify the regions of largest difference, and the marked boxes identify the enlarged intervals. The branch dependence grows from A to C but remains a small correction to the inclination-driven image structure.
\begin{figure}[tbp]
\centering
\includegraphics[width=0.98\textwidth]{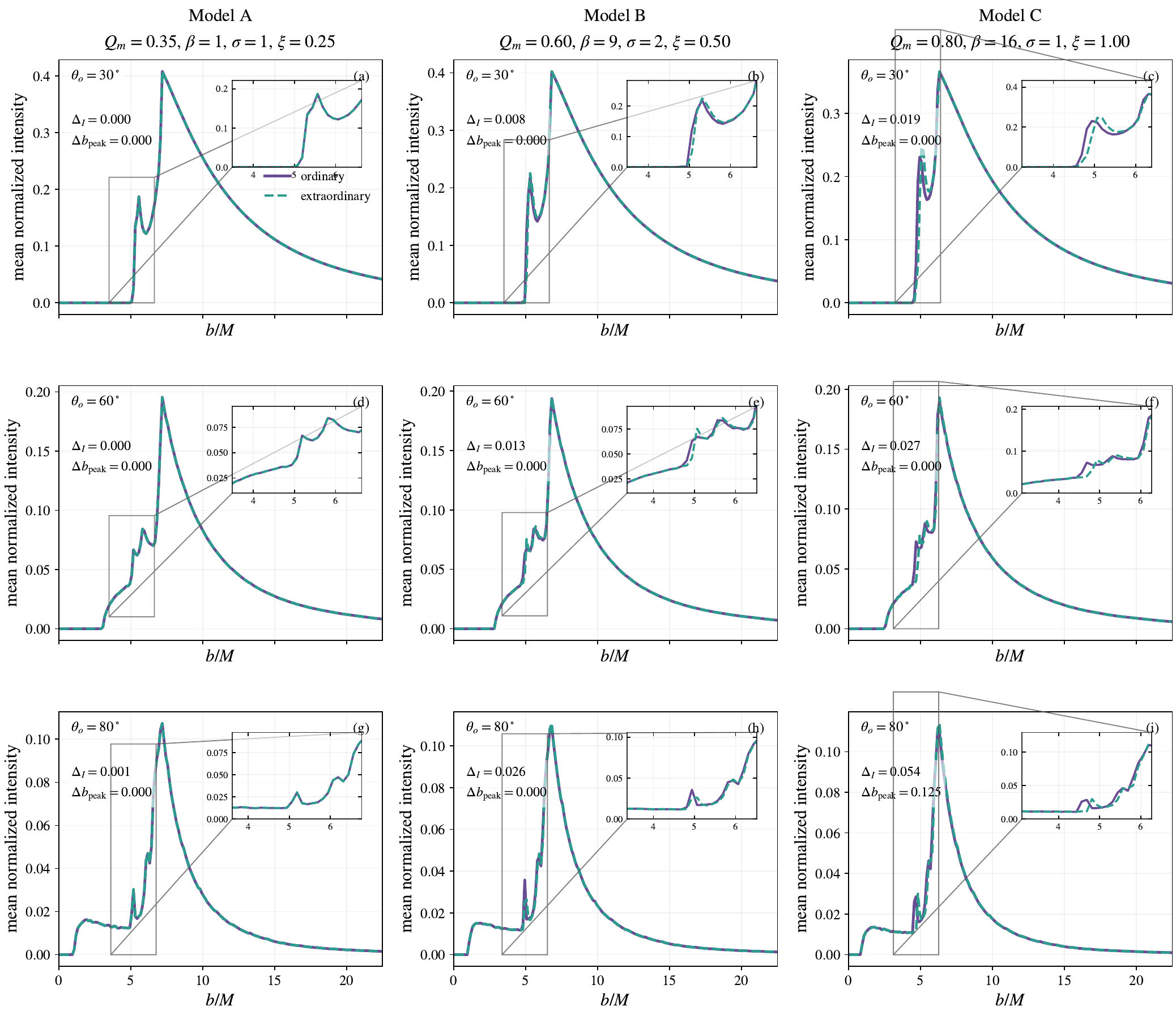}
\caption{Azimuthally averaged radial intensity profiles for the ordinary and extraordinary branches. Columns show models A--C and rows show the three inclinations. The profiles are obtained from the raw maps after division by the same joint peak intensity. Transparent insets magnify the region with the largest local difference; the corresponding interval is marked in each main panel.}
\label{fig:images_prof}
\end{figure}

The residual follows the lensed inner edge and the near-critical arcs rather than appearing as isolated pixels. A sign change across a bright feature is expected when the extraordinary branch shifts that feature slightly outward: one side loses intensity while the neighboring side gains it. The red--blue pairs therefore indicate a displacement or redistribution, not a new emission component. The flux ratio $R_F$ can remain close to unity while $\Delta_I$ is nonzero because positive and negative local changes partly cancel. Residual maps are therefore more informative than a direct visual comparison of the two full images.

The value of $\Delta b_{\rm peak}$ is limited by the radial bin width and can be zero when the two maxima fall in the same bin, even if the profiles differ inside the zoomed window. The critical impact parameters in Table~\ref{tab:optical} provide a cleaner measure of the geometric displacement, while the profiles show how it appears in the image intensity. Overall, the extraordinary branch does not create a different global disk morphology in these static spherical examples. Its main effect is a small redistribution of intensity near the lensed inner edge and an outward shift of the critical scale.

\FloatBarrier

\section{Discussion and conclusions}\label{sec:discussion}

We have introduced a class of NLED theories built from positive mixtures of concave power kernels. Finite electromagnetic self-energy requires $\gamma>1/4$, while causal magnetic propagation requires $\gamma\leq1/2$. Their intersection gives the interval $1/4<\gamma\leq1/2$. Because the mixture is positive, the relevant inequalities are preserved term by term. The pair $(\gamma_1,\gamma_2)=(1/3,1/2)$ is a minimal analytically tractable representative, not a unique choice.

For this two-kernel model, the static field equations admit an exact solution with finite electromagnetic self-energy and a singular center. The conditions $\LF>0$, $\LFF<0$, and $\Phi>0$ keep the extraordinary characteristic cone well defined and subluminal. The weak, dominant, and strong energy conditions also hold, together with the known sufficient stability conditions in the exterior \cite{MorenoSarbach2003,NomuraYoshidaSoda2020,RussoTownsend2024}. The central analytical result is the strict monotonicity of the metric function for $M_0\geq0$. In this sector, the solution has at most one positive-radius horizon, no inner Cauchy horizon, and no extremal horizon at finite radius. Thus magnetic charge does not force a Reissner--Nordstr\"om-like interior, even though the weak-field asymptotics recover the expected charged form \cite{HaleEtAl2026}.

The pure boundary branch displays the change of causal structure directly. Above the critical charge-to-scale ratio, the spacetime has one horizon and a spacelike singularity. Below it, the horizon disappears and the singularity is timelike and naked. At the threshold, the singularity becomes null. The center remains a curvature singularity throughout this transition; this one-horizon boundary family has no Cauchy horizon but does not regularize the central region \cite{DeFeliceTsujikawa2025,BokulicEtAl2026,RussoTownsend2026}.

The thermodynamic results are consistent with this geometric picture. The Hawking temperature is positive for every horizon at finite radius and tends to zero as the boundary branch approaches its critical endpoint. The representative fixed-coupling families have negative fixed-charge heat capacity and are therefore locally unstable in the asymptotically flat canonical ensemble. The extended first law and Smarr relation consistently include the magnetic charge and nonlinear couplings \cite{ZhangGao2018,BokulicThermo2021,AbeEtAl2026}.

The optical sector also shows why the spacetime metric and the electromagnetic effective metric must be treated separately. Magnetic charge reduces the ordinary critical scale, while the extraordinary branch shifts the critical curve outward and partly compensates for this reduction. The two branches also have different photon-orbit radii and instability rates. The simple shadow scan is only illustrative, but it shows that the effective optical branch can change the interpretation of an EHT-sized critical curve \cite{DePaulaEtAl2026,CimdikerEtAl2026}. In the thin-disk proxy images, inclination controls the global morphology. The extraordinary branch instead produces a small, coherent redistribution of intensity near the lensed inner edge and the critical region. Residual maps and radial profiles reveal this effect without identifying the critical curve with the bright lensing structures of the disk \cite{Luminet1979,GrallaEtAl2019}.

The present analysis is limited to static spherical symmetry, sufficient exterior stability criteria, an illustrative shadow-size comparison, and a simplified emission model. A direct calculation of the coupled gravitational--electromagnetic perturbation spectrum would provide a stronger stability test \cite{MorenoSarbach2003,NomuraYoshidaSoda2020,DaghighEtAl2022,FathiGuzmanVillanueva2026,LiangEtAl2026}. It would also be useful to study curvature components in parallel-propagated frames and to derive a rotating counterpart directly from the field equations \cite{GarciaDiaz2022,AyonBeato2024,ChengEtAl2025,HaleRotating2026}. More realistic observational tests will require polarized radiative transfer and a statistical treatment of the source parameters \cite{EHTSgrAVI2022,DePaulaEtAl2026}.

The main contribution is a constructive link between physical restrictions on an NLED Lagrangian and exact geometric consequences. The model gives a single framework for studying electromagnetic causality, global structure, thermodynamics, and birefringent optics.

\medskip
\noindent\textbf{Acknowledgements.}
M.F. acknowledges financial support from the Agencia Nacional de Investigaci\'{o}n y Desarrollo (ANID), Chile, through the FONDECYT Postdoctoral project No. 3260029. J.R.V. is partially supported by the Centro de F\'isica Te\'orica de Valpara\'iso (CeFiTeV).

\medskip
\noindent\textbf{Competing interests.}
The authors declare that they have no competing interests.

\medskip
\noindent\textbf{Data availability statement.}
No observational dataset was generated or reanalyzed in this work. Published EHT size estimates were used only to define an illustrative shadow-size band; no statistical fit was performed. The numerical data underlying the tables and figures, including the optical table, raw intensity arrays, residual maps, radial profiles, and image-comparison quantities, are available from the corresponding author upon reasonable request.

\medskip
\noindent\textbf{Code availability statement.}
The Python notebooks and scripts used for the analytical checks, numerical calculations, shadow scans, photon trajectories, and image analysis are available from the corresponding author upon reasonable request. The numerical work used standard open-source Python tools \cite{Meurer2017,Harris2020,Virtanen2020,Hunter2007}. The image code reproduces the jointly scaled images, residual maps, radial profiles, and numerical comparison quantities from the raw intensity arrays.

\FloatBarrier
\appendix

\section{Self-energy integral}\label{app:energy}

For one kernel, set $t=b^4/r^4$. Then
\begin{align}
 I={}&\int_0^\infty r^2
 \left[\left(1+\frac{b^4}{r^4}\right)^{1-\gamma}-1\right]\dd r
 \nonumber\\
 ={}&\frac{b^3}{4}\int_0^\infty
 t^{-7/4}\left[(1+t)^{1-\gamma}-1\right]\dd t.
\end{align}
For $\gamma>1/4$, integration by parts gives
\begin{align}
 I={}&\frac{b^3(1-\gamma)}{3}
 \int_0^\infty t^{-3/4}(1+t)^{-\gamma}\dd t
 \nonumber\\
 ={}&\frac{b^3(1-\gamma)}{3}
 B\!\left(\frac14,\gamma-\frac14\right),
\end{align}
where the last equality follows from Euler's beta integral \cite{DLMF,GradshteynRyzhik2015}.

\section{A simple check of the Smarr formula}\label{app:smarr}

Under the scaling
\begin{equation}
 r_+\to\lambda r_+,
 \qquad Q_m\to\lambda Q_m,
 \qquad \beta\to\lambda^2\beta,
\end{equation}
we obtain $M\to\lambda M$ and $S\to\lambda^2S$. The parameters $\sigma$ and $\xi$ remain unchanged. Euler's theorem then gives
\begin{equation}
 M=2S\frac{\partial M}{\partial S}
 +Q_m\frac{\partial M}{\partial Q_m}
 +2\beta\frac{\partial M}{\partial\beta},
\end{equation}
Using the first law, this becomes Eq.~\eqref{eq:Smarr}.

\bibliography{refs}

@article{Born1933,
  author = {Born, Max},
  title = {Modified field equations with a finite radius of the electron},
  journal = {Nature},
  volume = {132},
  pages = {282},
  year = {1933},
  doi = {10.1038/132282a0}
}

@article{BornInfeld1934,
  author = {Born, Max and Infeld, Leopold},
  title = {Foundations of the new field theory},
  journal = {Proc. R. Soc. Lond. A},
  volume = {144},
  pages = {425--451},
  year = {1934},
  doi = {10.1098/rspa.1934.0059}
}

@article{HeisenbergEuler1936,
  author = {Heisenberg, Werner and Euler, Hans},
  title = {Folgerungen aus der {Dirac}schen Theorie des Positrons},
  journal = {Z. Phys.},
  volume = {98},
  pages = {714--732},
  year = {1936},
  doi = {10.1007/BF01343663}
}

@book{Plebanski1970,
  author = {Pleba\'nski, Jerzy},
  title = {Lectures on Non-Linear Electrodynamics},
  year = {1970},
  publisher = {NORDITA},
  address = {Copenhagen}
}

@article{Boillat1970,
  author = {Boillat, Guy},
  title = {Nonlinear electrodynamics: Lagrangians and equations of motion},
  journal = {J. Math. Phys.},
  volume = {11},
  pages = {941--951},
  year = {1970},
  doi = {10.1063/1.1665231}
}

@incollection{BialynickiBirula1984,
  author = {Bia\l{}ynicki-Birula, Iwo and Bia\l{}ynicka-Birula, Zofia},
  title = {Nonlinear effects in quantum electrodynamics},
  booktitle = {Quantum Electrodynamics},
  editor = {Kinoshita, T.},
  pages = {226--310},
  year = {1990},
  publisher = {World Scientific}
}

@article{NovelloEtAl2000,
  author = {Novello, M. and De Lorenci, V. A. and Salim, J. M. and Klippert, R.},
  title = {Geometrical aspects of light propagation in nonlinear electrodynamics},
  journal = {Phys. Rev. D},
  volume = {61},
  pages = {045001},
  year = {2000},
  doi = {10.1103/PhysRevD.61.045001},
  eprint = {gr-qc/9911085}
}

@article{ObukhovRubilar2002,
  author = {Obukhov, Yuri N. and Rubilar, Guillermo F.},
  title = {Fresnel analysis of wave propagation in nonlinear electrodynamics},
  journal = {Phys. Rev. D},
  volume = {66},
  pages = {024042},
  year = {2002},
  doi = {10.1103/PhysRevD.66.024042},
  eprint = {gr-qc/0204028}
}

@article{ShabadUsov2011,
  author = {Shabad, Anatoly E. and Usov, Vladimir V.},
  title = {Effective {Lagrangian} in nonlinear electrodynamics and its properties of causality and unitarity},
  journal = {Phys. Rev. D},
  volume = {83},
  pages = {105006},
  year = {2011},
  doi = {10.1103/PhysRevD.83.105006},
  eprint = {1101.2343}
}

@article{Schellstede2016,
  author = {Schellstede, Gerold O. and Perlick, Volker and L\"ammerzahl, Claus},
  title = {On causality in nonlinear vacuum electrodynamics of the {Pleba\'nski} class},
  journal = {Ann. Phys. (Berlin)},
  volume = {528},
  pages = {738--749},
  year = {2016},
  doi = {10.1002/andp.201600124},
  eprint = {1604.02545}
}

@article{GibbonsHerdeiro2001,
  author = {Gibbons, G. W. and Herdeiro, C. A. R.},
  title = {{Born--Infeld} theory and stringy causality},
  journal = {Phys. Rev. D},
  volume = {63},
  pages = {064006},
  year = {2001},
  doi = {10.1103/PhysRevD.63.064006},
  eprint = {hep-th/0008052}
}

@article{RussoTownsend2024,
  author = {Russo, Jorge G. and Townsend, Paul K.},
  title = {Causality and energy conditions in nonlinear electrodynamics},
  journal = {J. High Energy Phys.},
  volume = {06},
  number = {2024},
  eid = {191},
  year = {2024},
  doi = {10.1007/JHEP06(2024)191},
  eprint = {2404.09994}
}

@article{AbeEtAl2026,
  author = {Abe, Yoshihiko and M\'edevielle, Maxime and Noumi, Toshifumi and Yoshimura, Kaho},
  title = {Causality constraints on black hole thermodynamics in nonlinear electrodynamics},
  journal = {J. High Energy Phys.},
  volume = {06},
  number = {2026},
  eid = {007},
  year = {2026},
  doi = {10.1007/JHEP06(2026)007},
  eprint = {2505.23483}
}

@article{RussoTownsend2026,
  author = {Russo, Jorge G. and Townsend, Paul K.},
  title = {Black holes and causal nonlinear electrodynamics},
  journal = {arXiv e-prints},
  year = {2026},
  eprint = {2601.07789}
}

@article{Sorokin2022,
  author = {Sorokin, Dmitri P.},
  title = {Introductory notes on non-linear electrodynamics and its applications},
  journal = {Fortschr. Phys.},
  volume = {70},
  pages = {2200092},
  year = {2022},
  doi = {10.1002/prop.202200092},
  eprint = {2112.12118}
}

@article{Bronnikov2001,
  author = {Bronnikov, Kirill A.},
  title = {Regular magnetic black holes and monopoles from nonlinear electrodynamics},
  journal = {Phys. Rev. D},
  volume = {63},
  pages = {044005},
  year = {2001},
  doi = {10.1103/PhysRevD.63.044005},
  eprint = {gr-qc/0006014}
}

@incollection{Bronnikov2022,
  author = {Bronnikov, Kirill A.},
  title = {Regular black holes sourced by nonlinear electrodynamics},
  booktitle = {Regular Black Holes: Towards a New Paradigm of Gravitational Collapse},
  editor = {Bambi, Cosimo},
  series = {Springer Series in Astrophysics and Cosmology},
  pages = {37--67},
  year = {2023},
  publisher = {Springer},
  address = {Singapore},
  doi = {10.1007/978-981-99-1596-5_2},
  eprint = {2211.00743}
}

@article{AyonBeatoGarcia1998,
  author = {Ay\'on-Beato, Eloy and Garc\'ia, Alberto},
  title = {Regular black hole in general relativity coupled to nonlinear electrodynamics},
  journal = {Phys. Rev. Lett.},
  volume = {80},
  pages = {5056--5059},
  year = {1998},
  doi = {10.1103/PhysRevLett.80.5056},
  eprint = {gr-qc/9911046}
}

@article{AyonBeatoGarcia1999,
  author = {Ay\'on-Beato, Eloy and Garc\'ia, Alberto},
  title = {New regular black hole solution from nonlinear electrodynamics},
  journal = {Phys. Lett. B},
  volume = {464},
  pages = {25--29},
  year = {1999},
  doi = {10.1016/S0370-2693(99)01038-2},
  eprint = {hep-th/9911174}
}

@article{AyonBeatoGarcia2000,
  author = {Ay\'on-Beato, Eloy and Garc\'ia, Alberto},
  title = {The {Bardeen} model as a nonlinear magnetic monopole},
  journal = {Phys. Lett. B},
  volume = {493},
  pages = {149--152},
  year = {2000},
  doi = {10.1016/S0370-2693(00)01125-4},
  eprint = {gr-qc/0009077}
}

@article{FanWang2016,
  author = {Fan, Zhong-Ying and Wang, Xiaobao},
  title = {Construction of regular black holes in general relativity},
  journal = {Phys. Rev. D},
  volume = {94},
  pages = {124027},
  year = {2016},
  doi = {10.1103/PhysRevD.94.124027},
  eprint = {1610.02636}
}

@article{BronnikovComment2017,
  author = {Bronnikov, Kirill A.},
  title = {Comment on ``Construction of regular black holes in general relativity''},
  journal = {Phys. Rev. D},
  volume = {96},
  pages = {128501},
  year = {2017},
  doi = {10.1103/PhysRevD.96.128501},
  eprint = {1712.04342}
}

@article{BokulicEtAl2024,
  author = {Bokuli\'c, Ana and Franzin, Edgardo and Juri\'c, Tajron and Smoli\'c, Ivica},
  title = {Lagrangian reverse engineering for regular black holes},
  journal = {Phys. Lett. B},
  volume = {853},
  pages = {138750},
  year = {2024},
  doi = {10.1016/j.physletb.2024.138750},
  eprint = {2311.17151}
}

@article{BokulicEtAl2026,
  author = {Bokuli\'c, Ana and Juri\'c, Tajron and Smoli\'c, Ivica},
  title = {Conundrum of regular black holes with nonlinear electromagnetic fields},
  journal = {Phys. Rev. D},
  volume = {113},
  pages = {024044},
  year = {2026},
  doi = {10.1103/z7gd-96ms},
  eprint = {2510.23711}
}

@article{TsudaEtAl2026,
  author = {Tsuda, Ren and Suzuki, Ryotaku and Tomizawa, Shinya},
  title = {Existence conditions of nonsingular dyonic black holes in nonlinear electrodynamics},
  journal = {Phys. Scr.},
  volume = {101},
  pages = {105005},
  year = {2026},
  doi = {10.1088/1402-4896/ae4b97},
  eprint = {2308.02146}
}

@article{DeFeliceTsujikawa2025,
  author = {De Felice, Antonio and Tsujikawa, Shinji},
  title = {Instability of nonsingular black holes in nonlinear electrodynamics},
  journal = {Phys. Rev. Lett.},
  volume = {134},
  pages = {081401},
  year = {2025},
  doi = {10.1103/PhysRevLett.134.081401},
  eprint = {2410.00314}
}

@article{MorenoSarbach2003,
  author = {Moreno, Claudia and Sarbach, Olivier},
  title = {Stability properties of black holes in self-gravitating nonlinear electrodynamics},
  journal = {Phys. Rev. D},
  volume = {67},
  pages = {024028},
  year = {2003},
  doi = {10.1103/PhysRevD.67.024028},
  eprint = {gr-qc/0208090}
}

@article{NomuraYoshidaSoda2020,
  author = {Nomura, Kosuke and Yoshida, Daisuke and Soda, Jiro},
  title = {Stability of magnetic black holes in general nonlinear electrodynamics},
  journal = {Phys. Rev. D},
  volume = {101},
  pages = {124026},
  year = {2020},
  doi = {10.1103/PhysRevD.101.124026},
  eprint = {2004.07560}
}

@article{NomuraYoshida2022,
  author = {Nomura, Kosuke and Yoshida, Daisuke},
  title = {Quasinormal modes of charged black holes with corrections from nonlinear electrodynamics},
  journal = {Phys. Rev. D},
  volume = {105},
  pages = {044006},
  year = {2022},
  doi = {10.1103/PhysRevD.105.044006},
  eprint = {2111.06273}
}

@article{DaghighEtAl2022,
  author = {Daghigh, Ramin G. and Green, Michael D. and Morey, Garrett and Kunstatter, Gabor},
  title = {Gravitational and electromagnetic radiation from an electrically charged black hole in general nonlinear electrodynamics},
  journal = {Phys. Rev. D},
  volume = {105},
  pages = {024055},
  year = {2022},
  doi = {10.1103/PhysRevD.105.024055},
  eprint = {2106.04038}
}

@article{BokulicThermo2021,
  author = {Bokuli\'c, Ana and Juri\'c, Tajron and Smoli\'c, Ivica},
  title = {Black hole thermodynamics in the presence of nonlinear electromagnetic fields},
  journal = {Phys. Rev. D},
  volume = {103},
  pages = {124059},
  year = {2021},
  doi = {10.1103/PhysRevD.103.124059},
  eprint = {2102.06213}
}

@article{ZhangGao2018,
  author = {Zhang, Yang and Gao, Sijie},
  title = {First law and {Smarr} formula of black hole mechanics in nonlinear gauge theories},
  journal = {Class. Quantum Grav.},
  volume = {35},
  pages = {145007},
  year = {2018},
  doi = {10.1088/1361-6382/aac9d4},
  eprint = {1610.01237}
}

@article{Smarr1973,
  author = {Smarr, Larry},
  title = {Mass formula for {Kerr} black holes},
  journal = {Phys. Rev. Lett.},
  volume = {30},
  pages = {71--73},
  year = {1973},
  doi = {10.1103/PhysRevLett.30.71}
}

@article{HaleEtAl2026,
  author = {Hale, Tom\'a\v{s} and Hennigar, Robie A. and Kubiz\v{n}\'ak, David},
  title = {Excising {Cauchy} horizons with nonlinear electrodynamics},
  journal = {Phys. Rev. D},
  volume = {113},
  pages = {L061502},
  year = {2026},
  doi = {10.1103/x8c8-j85k},
  eprint = {2506.20802}
}

@article{BarriolaVilenkin1989,
  author = {Barriola, Manuel and Vilenkin, Alexander},
  title = {Gravitational field of a global monopole},
  journal = {Phys. Rev. Lett.},
  volume = {63},
  pages = {341--343},
  year = {1989},
  doi = {10.1103/PhysRevLett.63.341}
}

@article{Penrose1965,
  author = {Penrose, Roger},
  title = {Gravitational collapse and space-time singularities},
  journal = {Phys. Rev. Lett.},
  volume = {14},
  pages = {57--59},
  year = {1965},
  doi = {10.1103/PhysRevLett.14.57}
}

@article{GarciaDiaz2022,
  author = {Garc\'ia-D\'iaz, Alberto A.},
  title = {{AdS--dS} stationary rotating black hole exact solution within {Einstein}--nonlinear electrodynamics},
  journal = {Ann. Phys.},
  volume = {441},
  pages = {168880},
  year = {2022},
  doi = {10.1016/j.aop.2022.168880},
  eprint = {2201.10682}
}

@article{AyonBeato2024,
  author = {Ay\'on-Beato, Eloy},
  title = {Unveiling the electrodynamics of the first nonlinearly charged rotating black hole},
  journal = {Ann. Phys.},
  volume = {469},
  pages = {169771},
  year = {2024},
  doi = {10.1016/j.aop.2024.169771},
  eprint = {2203.12809}
}

@article{ChengEtAl2025,
  author = {Cheng, Lang and Wang, Peng},
  title = {Rotating black holes in {Einstein}--{Born--Infeld} theory},
  journal = {arXiv e-prints},
  year = {2025},
  eprint = {2507.00879}
}

@article{BandosEtAl2020,
  author = {Bandos, Igor and Lechner, Kurt and Sorokin, Dmitri and Townsend, Paul K.},
  title = {A non-linear duality-invariant conformal extension of {Maxwell}'s equations},
  journal = {Phys. Rev. D},
  volume = {102},
  pages = {121703},
  year = {2020},
  doi = {10.1103/PhysRevD.102.121703},
  eprint = {2007.09092}
}

@article{BarrientosEtAl2025,
  author = {Barrientos, Jos\'e and Cisterna, Adolfo and Hassa\"ine, Mokhtar and Pallikaris, Konstantinos},
  title = {Electromagnetized black holes and swirling backgrounds in nonlinear electrodynamics: The {ModMax} case},
  journal = {Phys. Lett. B},
  volume = {860},
  pages = {139214},
  year = {2025},
  doi = {10.1016/j.physletb.2024.139214},
  eprint = {2409.12336}
}

@article{GibbonsRasheed1995,
  author = {Gibbons, G. W. and Rasheed, D. A.},
  title = {Electric-magnetic duality rotations in nonlinear electrodynamics},
  journal = {Nucl. Phys. B},
  volume = {454},
  pages = {185--206},
  year = {1995},
  doi = {10.1016/0550-3213(95)00409-L},
  eprint = {hep-th/9506035}
}

@article{HaleRotating2026,
  author = {Hale, Tom\'a\v{s} and Hennigar, Robie A. and Kubiz\v{n}\'ak, David},
  title = {Rotating extremal black holes in {Einstein}--{Born--Infeld} theory},
  journal = {J. High Energy Phys.},
  volume = {01},
  number = {2026},
  eid = {155},
  year = {2026},
  doi = {10.1007/JHEP01(2026)155},
  eprint = {2509.13099}
}

@book{SchillingBernstein2012,
  author = {Schilling, Ren\'e L. and Song, Renming and Vondra\v{c}ek, Zoran},
  title = {Bernstein Functions: Theory and Applications},
  year = {2012},
  publisher = {De Gruyter},
  address = {Berlin},
  edition = {2},
  doi = {10.1515/9783110269338}
}

@misc{DLMF,
  author = {{NIST Digital Library of Mathematical Functions}},
  title = {Sec.~8.17: Incomplete Beta Functions},
  year = {2026},
  howpublished = {\url{https://dlmf.nist.gov/8.17}},
  note = {Accessed 29 July 2026}
}

@book{GradshteynRyzhik2015,
  author = {Gradshteyn, I. S. and Ryzhik, I. M.},
  title = {Table of Integrals, Series, and Products},
  year = {2015},
  publisher = {Academic Press},
  edition = {8}
}

@article{Bekenstein1973,
  author = {Bekenstein, Jacob D.},
  title = {Black holes and entropy},
  journal = {Phys. Rev. D},
  volume = {7},
  pages = {2333--2346},
  year = {1973},
  doi = {10.1103/PhysRevD.7.2333}
}

@article{Hawking1975,
  author = {Hawking, Stephen W.},
  title = {Particle creation by black holes},
  journal = {Commun. Math. Phys.},
  volume = {43},
  pages = {199--220},
  year = {1975},
  doi = {10.1007/BF02345020}
}

@article{Wald1993,
  author = {Wald, Robert M.},
  title = {Black hole entropy is the {Noether} charge},
  journal = {Phys. Rev. D},
  volume = {48},
  pages = {R3427--R3431},
  year = {1993},
  doi = {10.1103/PhysRevD.48.R3427},
  eprint = {gr-qc/9307038}
}

@article{Meurer2017,
  author = {Meurer, Aaron and Smith, Christopher P. and Paprocki, Mateusz and others},
  title = {{SymPy}: symbolic computing in {Python}},
  journal = {PeerJ Comput. Sci.},
  volume = {3},
  pages = {e103},
  year = {2017},
  doi = {10.7717/peerj-cs.103}
}

@article{Harris2020,
  author = {Harris, Charles R. and Millman, K. Jarrod and van der Walt, St\'efan J. and others},
  title = {Array programming with {NumPy}},
  journal = {Nature},
  volume = {585},
  pages = {357--362},
  year = {2020},
  doi = {10.1038/s41586-020-2649-2}
}

@article{Virtanen2020,
  author = {Virtanen, Pauli and Gommers, Ralf and Oliphant, Travis E. and others},
  title = {{SciPy} 1.0: fundamental algorithms for scientific computing in {Python}},
  journal = {Nat. Methods},
  volume = {17},
  pages = {261--272},
  year = {2020},
  doi = {10.1038/s41592-019-0686-2}
}

@article{Hunter2007,
  author = {Hunter, John D.},
  title = {Matplotlib: a {2D} graphics environment},
  journal = {Comput. Sci. Eng.},
  volume = {9},
  pages = {90--95},
  year = {2007},
  doi = {10.1109/MCSE.2007.55}
}

@article{Kruskal1960,
  author = {Kruskal, Martin D.},
  title = {Maximal extension of {Schwarzschild} metric},
  journal = {Phys. Rev.},
  volume = {119},
  pages = {1743--1745},
  year = {1960},
  doi = {10.1103/PhysRev.119.1743}
}

@article{Szekeres1960,
  author = {Szekeres, George},
  title = {On the singularities of a {Riemannian} manifold},
  journal = {Publ. Math. Debrecen},
  volume = {7},
  pages = {285--301},
  year = {1960}
}

@book{HawkingEllis1973,
  author = {Hawking, S. W. and Ellis, G. F. R.},
  title = {The Large Scale Structure of Space-Time},
  year = {1973},
  publisher = {Cambridge University Press},
  address = {Cambridge}
}

@book{Wald1984,
  author = {Wald, Robert M.},
  title = {General Relativity},
  year = {1984},
  publisher = {University of Chicago Press},
  address = {Chicago}
}

@article{Eddington1924,
  author = {Eddington, Arthur S.},
  title = {A comparison of {Whitehead}'s and {Einstein}'s formulae},
  journal = {Nature},
  volume = {113},
  pages = {192},
  year = {1924},
  doi = {10.1038/113192a0}
}

@article{Finkelstein1958,
  author = {Finkelstein, David},
  title = {Past-future asymmetry of the gravitational field of a point particle},
  journal = {Phys. Rev.},
  volume = {110},
  pages = {965--967},
  year = {1958},
  doi = {10.1103/PhysRev.110.965}
}

@book{FrolovNovikov1998,
  author = {Frolov, Valeri P. and Novikov, Igor D.},
  title = {Black Hole Physics: Basic Concepts and New Developments},
  year = {1998},
  publisher = {Kluwer Academic},
  address = {Dordrecht},
  doi = {10.1007/978-94-011-5139-9}
}

@article{Escobar2026,
  author = {Escobar, C. A.},
  title = {Effective-metric formulation of {Casimir} energies in nonlinear scalar and electromagnetic theories},
  journal = {Phys. Rev. D},
  volume = {114},
  pages = {025010},
  year = {2026},
  doi = {10.1103/hg8s-pxzc},
  eprint = {2606.17221}
}

@article{DePaulaEtAl2026,
  author = {de Paula, Marco A. A. and Lima, Haroldo C. D. and Cunha, Pedro V. P. and Herdeiro, Carlos A. R. and Crispino, Lu\'is C. B.},
  title = {Two shadows of a single black hole: Vacuum birefringence phenomena within {Einstein}--nonlinear-electrodynamics},
  journal = {Phys. Rev. D},
  volume = {114},
  pages = {024043},
  year = {2026},
  doi = {10.1103/8lkf-mgjw},
  eprint = {2603.17007}
}

@article{DePaulaEtAl2023,
  author = {de Paula, Marco A. A. and Lima Junior, Haroldo C. D. and Cunha, Pedro V. P. and Crispino, Lu\'is C. B.},
  title = {Electrically charged regular black holes in nonlinear electrodynamics: Light rings, shadows, and gravitational lensing},
  journal = {Phys. Rev. D},
  volume = {108},
  pages = {084029},
  year = {2023},
  doi = {10.1103/PhysRevD.108.084029},
  eprint = {2305.04776}
}

@article{UniyalEtAl2023,
  author = {Uniyal, Akhil and Chakrabarti, Sayan and Fathi, Mohsen and {\"O}vg\"un, Ali},
  title = {Observational signatures: Shadow cast by the effective metric of photons for black holes with rational non-linear electrodynamics},
  journal = {Ann. Phys.},
  volume = {462},
  pages = {169614},
  year = {2024},
  doi = {10.1016/j.aop.2024.169614},
  eprint = {2309.13680}
}

@article{PerlickEtAl2015,
  author = {Perlick, Volker and Tsupko, Oleg Yu. and Bisnovatyi-Kogan, Gennady S.},
  title = {Influence of a plasma on the shadow of a spherically symmetric black hole},
  journal = {Phys. Rev. D},
  volume = {92},
  pages = {104031},
  year = {2015},
  doi = {10.1103/PhysRevD.92.104031},
  eprint = {1507.04217}
}

@article{GrallaEtAl2019,
  author = {Gralla, Samuel E. and Holz, Daniel E. and Wald, Robert M.},
  title = {Black hole shadows, photon rings, and lensing rings},
  journal = {Phys. Rev. D},
  volume = {100},
  pages = {024018},
  year = {2019},
  doi = {10.1103/PhysRevD.100.024018},
  eprint = {1906.00873}
}

@article{EHTM87VI2019,
  author = {{Event Horizon Telescope Collaboration}},
  title = {First {M87} {Event Horizon Telescope} results. {VI}. {The} shadow and mass of the central black hole},
  journal = {Astrophys. J. Lett.},
  volume = {875},
  pages = {L6},
  year = {2019},
  doi = {10.3847/2041-8213/ab1141},
  eprint = {1906.11243}
}

@article{EHTSgrAVI2022,
  author = {{Event Horizon Telescope Collaboration}},
  title = {First {Sagittarius A*} {Event Horizon Telescope} results. {VI}. {Testing} the black hole metric},
  journal = {Astrophys. J. Lett.},
  volume = {930},
  pages = {L17},
  year = {2022},
  doi = {10.3847/2041-8213/ac6756}
}

@article{PageThorne1974,
  author = {Page, Don N. and Thorne, Kip S.},
  title = {Disk-accretion onto a black hole. {Time-averaged} structure of accretion disk},
  journal = {Astrophys. J.},
  volume = {191},
  pages = {499--506},
  year = {1974},
  doi = {10.1086/152990}
}

@article{Luminet1979,
  author = {Luminet, Jean-Pierre},
  title = {Image of a spherical black hole with thin accretion disk},
  journal = {Astron. Astrophys.},
  volume = {75},
  pages = {228--235},
  year = {1979}
}

@article{FathiGuzmanVillanueva2026,
  author = {Fathi, Mohsen and Guzm\'{a}n, Ariel and Villanueva, J. R.},
  title = {Quasinormal modes of a static black hole in nonlinear electrodynamics},
  journal = {Eur. Phys. J. C},
  volume = {86},
  pages = {587},
  year = {2026},
  doi = {10.1140/epjc/s10052-026-15713-0}
}

@article{AyonBeatoGarcia2005,
  author = {Ay\'{o}n-Beato, Eloy and Garc\'{i}a, Alberto},
  title = {Four-parametric regular black hole solution},
  journal = {Gen. Relativ. Gravit.},
  volume = {37},
  pages = {635--641},
  year = {2005},
  doi = {10.1007/s10714-005-0050-y},
  eprint = {hep-th/0403229}
}

@article{Dymnikova2004,
  author = {Dymnikova, Irina},
  title = {Regular electrically charged vacuum structures with {de Sitter} centre in nonlinear electrodynamics coupled to general relativity},
  journal = {Class. Quantum Grav.},
  volume = {21},
  pages = {4417--4428},
  year = {2004},
  doi = {10.1088/0264-9381/21/18/009},
  eprint = {gr-qc/0407072}
}

@article{Breton2003,
  author = {Bret\'{o}n, Nora},
  title = {{Born--Infeld} black hole in the isolated horizon framework},
  journal = {Phys. Rev. D},
  volume = {67},
  pages = {124004},
  year = {2003},
  doi = {10.1103/PhysRevD.67.124004},
  eprint = {hep-th/0301254}
}

@article{FernandoKrug2003,
  author = {Fernando, Sharmanthie and Krug, Don},
  title = {Charged black hole solutions in {Einstein}--{Born--Infeld} gravity with a cosmological constant},
  journal = {Gen. Relativ. Gravit.},
  volume = {35},
  pages = {129--137},
  year = {2003},
  doi = {10.1023/A:1021315214180},
  eprint = {hep-th/0306120}
}

@article{CaiPangWang2004,
  author = {Cai, Rong-Gen and Pang, Da-Wei and Wang, Anzhong},
  title = {{Born--Infeld} black holes in {(A)dS} spaces},
  journal = {Phys. Rev. D},
  volume = {70},
  pages = {124034},
  year = {2004},
  doi = {10.1103/PhysRevD.70.124034},
  eprint = {hep-th/0410158}
}

@article{Fernando2006,
  author = {Fernando, Sharmanthie},
  title = {Thermodynamics of {Born--Infeld}--anti-{de Sitter} black holes in the grand canonical ensemble},
  journal = {Phys. Rev. D},
  volume = {74},
  pages = {104032},
  year = {2006},
  doi = {10.1103/PhysRevD.74.104032},
  eprint = {hep-th/0608040}
}

@article{Ma2015,
  author = {Ma, Meng-Sen},
  title = {Magnetically charged regular black hole in a model of nonlinear electrodynamics},
  journal = {Ann. Phys.},
  volume = {362},
  pages = {529--537},
  year = {2015},
  doi = {10.1016/j.aop.2015.08.028},
  eprint = {1509.05580}
}

@article{Kruglov2017,
  author = {Kruglov, S. I.},
  title = {{Born--Infeld}-type electrodynamics and magnetic black holes},
  journal = {Ann. Phys.},
  volume = {383},
  pages = {550--559},
  year = {2017},
  doi = {10.1016/j.aop.2017.06.008},
  eprint = {1707.04495}
}

@article{OkyayOvgun2022,
  author = {Okyay, Mert and \"{O}vg\"{u}n, Ali},
  title = {Nonlinear electrodynamics effects on the black hole shadow, deflection angle, quasinormal modes and greybody factors},
  journal = {J. Cosmol. Astropart. Phys.},
  volume = {01},
  number = {2022},
  eid = {009},
  year = {2022},
  doi = {10.1088/1475-7516/2022/01/009},
  eprint = {2108.07766}
}

@article{UniyalPantigOvgun2023,
  author = {Uniyal, Akhil and Pantig, Reggie C. and \"{O}vg\"{u}n, Ali},
  title = {Probing a nonlinear electrodynamics black hole with thin accretion disk, shadow, and deflection angle with {M87*} and {Sgr A*} from {EHT}},
  journal = {Phys. Dark Univ.},
  volume = {40},
  pages = {101178},
  year = {2023},
  doi = {10.1016/j.dark.2023.101178},
  eprint = {2205.11072}
}

@article{UniyalEtAl2024,
  author = {Uniyal, Akhil and Chakrabarti, Sayan and Pantig, Reggie C. and \"{O}vg\"{u}n, Ali},
  title = {Nonlinearly charged black holes: Shadow and thin-accretion disk},
  journal = {New Astron.},
  volume = {111},
  pages = {102249},
  year = {2024},
  doi = {10.1016/j.newast.2024.102249},
  eprint = {2303.07174}
}

@article{LambiaseEtAl2025,
  author = {Lambiase, Gaetano and Gogoi, Dhruba Jyoti and Pantig, Reggie C. and \"{O}vg\"{u}n, Ali},
  title = {Shadow and quasinormal modes of the rotating {Einstein}--{Euler--Heisenberg} black holes},
  journal = {Phys. Dark Univ.},
  volume = {48},
  pages = {101886},
  year = {2025},
  doi = {10.1016/j.dark.2025.101886},
  eprint = {2406.18300}
}

@article{VerbinEtAl2025,
  author = {Verbin, Yosef and Puli\c{c}e, Beyhan and \"{O}vg\"{u}n, Ali and Huang, Hyat},
  title = {New black hole solutions of second and first order formulations of nonlinear electrodynamics},
  journal = {Phys. Rev. D},
  volume = {111},
  pages = {084061},
  year = {2025},
  doi = {10.1103/PhysRevD.111.084061},
  eprint = {2412.20989}
}

@article{CimdikerEtAl2026,
  author = {{\c{C}}imdiker, \`Ilim \`Irfan and \"{O}vg\"{u}n, Ali and Verbin, Yosef},
  title = {Optical and orbital characterization of spherically symmetric static black holes of self-gravitating new nonlinear electrodynamics model},
  journal = {Phys. Rev. D},
  volume = {114},
  pages = {024032},
  year = {2026},
  doi = {10.1103/p6lp-rktr},
  eprint = {2603.10097}
}

@article{UktamovEtAl2026,
  author = {Uktamov, Uktamjon and \"{O}vg\"{u}n, Ali and Pantig, Reggie C. and Ahmedov, Bobomurat},
  title = {Horizon-brightened acceleration radiation and optical signatures of generic regular black holes from nonlinear electrodynamics},
  journal = {Eur. Phys. J. C},
  volume = {86},
  pages = {631},
  year = {2026},
  doi = {10.1140/epjc/s10052-026-15841-7},
  eprint = {2602.15077}
}

@article{DePaulaEtAl2024,
  author = {de Paula, Marco A. A. and Leite, Luiz C. S. and Dolan, Sam R. and Crispino, Lu\'{i}s C. B.},
  title = {Absorption and unbounded superradiance in a static regular black hole spacetime},
  journal = {Phys. Rev. D},
  volume = {109},
  pages = {064053},
  year = {2024},
  doi = {10.1103/PhysRevD.109.064053},
  eprint = {2401.01767}
}

@article{LeeMyung2026,
  author = {Lee, Wonwoo and Myung, Yun Soo},
  title = {Anisotropic matter and nonlinear electromagnetics black holes},
  journal = {Eur. Phys. J. C},
  volume = {86},
  pages = {525},
  year = {2026},
  doi = {10.1140/epjc/s10052-026-15746-5}
}

@article{LiangEtAl2026,
  author = {Liang, Jie and Liu, Dong and Long, Zheng-Wen},
  title = {Quasinormal modes and greybody factors of black holes corrected by nonlinear electrodynamics},
  journal = {Eur. Phys. J. C},
  volume = {86},
  pages = {17},
  year = {2026},
  doi = {10.1140/epjc/s10052-025-15245-z}
}

@article{CroneyEtAl2025,
  author = {Croney, Lewis and Gregory, Ruth and Ram\'{i}rez-Valdez, Carlos J.},
  title = {Thermodynamics of dyonic black holes in non-linear electrodynamics},
  journal = {J. High Energy Phys.},
  volume = {10},
  number = {2025},
  eid = {013},
  year = {2025},
  doi = {10.1007/JHEP10(2025)013},
  eprint = {2506.06437}
}

@article{ChenEtAl2026,
  author = {Chen, Che-Yu and De Felice, Antonio and Tsujikawa, Shinji and Sano, Taishi},
  title = {Vector {Horndeski} black holes in nonlinear electrodynamics},
  journal = {Phys. Rev. D},
  volume = {113},
  pages = {024027},
  year = {2026},
  doi = {10.1103/fjqh-7gb2}
}

@article{AnLiYang2021,
  author = {An, Yu-Sen and Li, Li and Yang, Fu-Guo},
  title = {No {Cauchy} horizon theorem for nonlinear electrodynamics black holes with charged scalar hairs},
  journal = {Phys. Rev. D},
  volume = {104},
  pages = {024040},
  year = {2021},
  doi = {10.1103/PhysRevD.104.024040},
  eprint = {2106.01069}
}

@article{BokulicHerdeiro2025,
  author = {Bokuli\'{c}, Ana and Herdeiro, Carlos A. R.},
  title = {Exact multiblack hole spacetimes in {Einstein}--{ModMax} theory},
  journal = {Phys. Rev. D},
  volume = {111},
  pages = {064046},
  year = {2025},
  doi = {10.1103/PhysRevD.111.064046}
}

@article{Bronnikov2024,
  author = {Bronnikov, Kirill A.},
  title = {Regular black holes as an alternative to black bounce},
  journal = {Phys. Rev. D},
  volume = {110},
  pages = {024021},
  year = {2024},
  doi = {10.1103/PhysRevD.110.024021}
}

@article{AlencarEtAl2024,
  author = {Alencar, G. and Bronnikov, Kirill A. and Rodrigues, Manuel E. and S\'{a}ez-Chill\'{o}n G\'{o}mez, Diego and Silva, Marcos V. de S.},
  title = {On black bounce space-times in non-linear electrodynamics},
  journal = {Eur. Phys. J. C},
  volume = {84},
  pages = {745},
  year = {2024},
  doi = {10.1140/epjc/s10052-024-13119-4}
}

@article{PinedoFrolov2026,
  author = {Pinedo Soto, Jose and Frolov, Valeri P.},
  title = {Charged black holes in quasitopological gravity coupled to {Born--Infeld} nonlinear electrodynamics},
  journal = {Phys. Rev. D},
  volume = {113},
  pages = {124044},
  year = {2026},
  doi = {10.1103/8px4-2pyk}
}

@article{RazaEtAl2024,
  author = {Raza, Muhammad Ali and Rayimbaev, Javlon and Sarikulov, Furkat and Zubair, M. and Ahmedov, Bobomurat and Stuchl\'{i}k, Zden\v{e}k},
  title = {Shadow of novel rotating black hole in general relativity coupled to nonlinear electrodynamics and constraints from {EHT} results},
  journal = {Phys. Dark Univ.},
  volume = {44},
  pages = {101488},
  year = {2024},
  doi = {10.1016/j.dark.2024.101488},
  eprint = {2311.15784}
}

@article{SarkarEtAl2025,
  author = {Sarkar, Susmita and Sarkar, Nayan and Shah, Hasrat Hussian and Balo, Pankaj and Rahaman, Farook},
  title = {Deflection angle of regular black holes in nonlinear electrodynamics: {Gauss--Bonnet} theorem, time delay, shadow, and greybody bound},
  journal = {Phys. Lett. B},
  volume = {870},
  pages = {139905},
  year = {2025},
  doi = {10.1016/j.physletb.2025.139905}
}

@article{SaleemEtAl2026,
  author = {Saleem, Amna and Majeed, Bushra and Ali, Zulfiqar and Ditta, Allah and Alimova, Asalkhon and Channuie, Phongpichit and Atamurotov, Farruh},
  title = {Impact of nonlinear electrodynamics on particle motion around a charged black hole with matter coupling},
  journal = {Eur. Phys. J. C},
  volume = {86},
  pages = {7},
  year = {2026},
  doi = {10.1140/epjc/s10052-025-15166-x}
}

@article{LiLu2024,
  author = {Li, Zhi-Chao and L\"{u}, Hong},
  title = {Regular electric black holes from {Einstein}--{Maxwell}--scalar gravity},
  journal = {Phys. Rev. D},
  volume = {110},
  pages = {104046},
  year = {2024},
  doi = {10.1103/PhysRevD.110.104046}
}

@article{ContrerasEtAl2025,
  author = {Contreras, Ernesto and Carrasco-Hidalgo, Mikaela and Bargue\~{n}o, Pedro and Suvorov, Arthur G.},
  title = {General framework for the spontaneous scalarization of regular black holes},
  journal = {Phys. Rev. D},
  volume = {112},
  pages = {124053},
  year = {2025},
  doi = {10.1103/9st9-ys67}
}

@article{BalartFernando2021,
  author = {Balart, Leonardo and Fernando, Sharmanthie},
  title = {Thermodynamics and heat engines of black holes with {Born--Infeld}-type electrodynamics},
  journal = {Mod. Phys. Lett. A},
  volume = {36},
  pages = {2150102},
  year = {2021},
  doi = {10.1142/S0217732321501029},
  eprint = {2103.15040}
}

@article{CaiMiao2021,
  author = {Cai, Xin-Chang and Miao, Yan-Gang},
  title = {Quasinormal modes and shadows of a new family of {Ay\'{o}n-Beato--Garc\'{i}a} black holes},
  journal = {Phys. Rev. D},
  volume = {103},
  pages = {124050},
  year = {2021},
  doi = {10.1103/PhysRevD.103.124050},
  eprint = {2104.09725}
}

@article{HegdeEtAl2025,
  author = {Hegde, Kartheek and Kumara, A. Naveena and Rizwan, C. L. Ahmed and Ali, Md Sabir and Ajith, K. M.},
  title = {Thermodynamics, photon sphere and thermodynamic geometry of {Ay\'{o}n-Beato--Garc\'{i}a} spacetime},
  journal = {Int. J. Mod. Phys. A},
  volume = {40},
  pages = {2550115},
  year = {2025},
  doi = {10.1142/S0217751X25501155},
  eprint = {2104.08091}
}

@article{AyonBeatoFloresHassaine2024,
  author = {Ay\'{o}n-Beato, Eloy and Flores-Alfonso, Daniel and Hassaine, Mokhtar},
  title = {Nonlinearly charging the conformally dressed black holes preserving duality and conformal invariance},
  journal = {Phys. Rev. D},
  volume = {110},
  pages = {064027},
  year = {2024},
  doi = {10.1103/PhysRevD.110.064027},
  eprint = {2404.08753}
}

@article{Dey2004,
  author = {Dey, Tanay Kumar},
  title = {{Born--Infeld} black holes in the presence of a cosmological constant},
  journal = {Phys. Lett. B},
  volume = {595},
  pages = {484--490},
  year = {2004},
  doi = {10.1016/j.physletb.2004.06.047},
  eprint = {hep-th/0406169}
}

@article{CataldoGarcia2000,
  author = {Cataldo, Mauricio and Garc\'{i}a, Alberto},
  title = {Regular (2+1)-dimensional black holes within nonlinear electrodynamics},
  journal = {Phys. Rev. D},
  volume = {61},
  pages = {084003},
  year = {2000},
  doi = {10.1103/PhysRevD.61.084003}
}

@article{RodriguesSilvaSiqueira2020,
  author = {Rodrigues, Manuel E. and Silva, Marcos V. de S. and de Siqueira, Andrew S.},
  title = {Regular multihorizon black holes in general relativity},
  journal = {Phys. Rev. D},
  volume = {102},
  pages = {084038},
  year = {2020},
  doi = {10.1103/PhysRevD.102.084038}
}

@article{DeFeliceTsujikawaScalar2025,
  author = {De Felice, Antonio and Tsujikawa, Shinji},
  title = {Nonsingular black holes and spherically symmetric objects in nonlinear electrodynamics with a scalar field},
  journal = {Phys. Rev. D},
  volume = {111},
  pages = {064051},
  year = {2025},
  doi = {10.1103/PhysRevD.111.064051}
}

@article{WangWuYang2019,
  author = {Wang, Peng and Wu, Houwen and Yang, Haitang},
  title = {Thermodynamics of nonlinear electrodynamics black holes and the validity of weak cosmic censorship at charged particle absorption},
  journal = {Eur. Phys. J. C},
  volume = {79},
  pages = {572},
  year = {2019},
  doi = {10.1140/epjc/s10052-019-7090-z}
}

\end{document}